\documentclass[a4paper,onecolumn,11pt,unpublished]{quantumarticle}
\pdfoutput=1
\usepackage[utf8]{inputenc}
\usepackage[english]{babel}
\usepackage[T1]{fontenc}

\usepackage{enumerate}
\usepackage{hyperref}
\usepackage{graphicx,color}
\usepackage{cite}
\usepackage{amsthm, amsmath, amssymb, bbm, mathtools, braket, multirow, url, hhline}

\renewcommand{\epsilon}{\varepsilon}
\newcommand{\abs}[1]{\left|#1\right|}

\renewcommand{\Im}{\operatorname{Im}}
\DeclareMathOperator{\Tr}{Tr}
\newcommand{\ketbra}[2]{\ket{#1}\!\bra{#2}}
\DeclareMathOperator{\diag}{diag}
\DeclareMathOperator*{\avg}{avg}

\newtheorem{theorem}{Theorem}[section]
\newtheorem{lemma}[theorem]{Lemma}
\newtheorem{corollary}[theorem]{Corollary}
\newtheorem{definition}[theorem]{Definition}
\newtheorem{proposition}[theorem]{Proposition}

\theoremstyle{definition}
\newtheorem{remark}[theorem]{Remark}

\newcommand{\cA}{\mathcal{A}}
\newcommand{\cB}{\mathcal{B}}
\newcommand{\cC}{\mathcal{C}}

\newcommand{\cH}{\mathcal{H}}

\newcommand{\cK}{\mathcal{K}}

\newcommand{\cQ}{\mathcal{Q}}

\newcommand{\cS}{\mathcal{S}}

\newcommand{\cX}{\mathcal{X}}
\newcommand{\cY}{\mathcal{Y}}

\newcommand{\C}{\mathrm{C}}
\newcommand{\Q}{\mathrm{Q}}
\newcommand{\OPT}{\mathrm{OPT}}
\newcommand{\SLD}{\mathrm{SLD}}
\newcommand{\RLD}{\mathrm{RLD}}
\newcommand{\BKM}{\mathrm{BKM}}
\newcommand{\WYD}{\mathrm{WYD}}
\newcommand{\mix}{\mathrm{mix}}
\newcommand{\Diag}{\mathsf{D}}
\newcommand{\sym}{\mathrm{sym}}
\newcommand{\asym}{\mathrm{asym}}
\newcommand{\iso}{\mathrm{iso}}

\begin{document}

\title{Optimal quantum locally differentially private mechanisms in the high-privacy regime}

\author{Yuuya Yoshida}
\email{yyoshida9130@gmail.com}
\affiliation{Independent Researcher, Komaki, Aichi, 485-0003, Japan}
\orcid{0000-0003-1071-2098}
\maketitle

\begin{abstract}
We optimize the trade-off between privacy and utility in the high-privacy regime. We adopt local differential privacy (LDP) and its quantum extension, quantum local differential privacy (QLDP), for privacy protection, and investigate utility functions including the Holevo information (which reduces to the mutual information in the classical case) and the error exponents in symmetric and asymmetric hypothesis testing. Under the LDP and QLDP constraints, these utility functions have classical and quantum optimal values, which we denote simply by $C$ and $Q$, respectively, in this abstract. In this paper, we provide optimal LDP and QLDP mechanisms achieving the classical and quantum optimal values in the high-privacy regime, and prove that the asymptotic ratio $Q/C$ in this regime takes the same value regardless of the utility function. Our results reveal quantum advantages (specifically, $Q/C\ge3/2$) for the above utility functions when the private data take at least three possible values. In addition, conditional on a mathematical conjecture concerning equi-isoclinic tight fusion frames (EITFFs), our QLDP mechanisms also exhibit quantum advantages in a moderate-privacy regime.
\end{abstract}

\section{Introduction}\label{intro}
With the rapid advancement and the increasing pervasiveness of information technology in daily life, 
it has become increasingly important to utilize private data while ensuring their protection.
Data utility and privacy protection are generally in a trade-off relation, and this trade-off needs to be optimized.
In particular, differential privacy (DP) has attracted attention as a rigorous framework for privacy protection, and 
various optimization problems under the DP constraint have been extensively studied 
\cite{Kairouz,Holohan1, Geng1,Geng2,Geng3,Geng4, ISIT2020,Yoshida,Nam, Guan,Nuradha}.
Originally, DP was introduced in the context of global privacy \cite{Dwork1,Dwork2}, 
where a trusted data collector (e.g., a corporation or government) releases statistical results or synthetic data based on users' data.
Since then, DP has also been widely applied to privacy-preserving machine learning.
Subsequently, the concept of DP was extended to the context of local privacy \cite{Duchi}, 
where data providers do not trust the data collector.
This paper focuses on the local privacy setting.
Specifically, the data providers convert their private data $X$, which takes values in a finite set $\cX$, 
into randomized data $Y$ taking values in another finite set $\cY$, and release $Y$.
Based on the released data $Y$, the data collector estimates certain information about $X$, such as its probability distribution or its expectation.
The conversion is specified by a conditional probability distribution $\mathbb{P}_{Y|X}$.
As stated later, local differential privacy (LDP) is a condition for conditional probability distributions, and 
its privacy loss bound is represented by a positive parameter $\epsilon$.

As stated above, optimizing the trade-off between privacy and utility is a central line of research in DP.
More specifically, the objective is to maximize the estimation accuracy of the data collector while simultaneously protecting the private data $X$ via LDP.
If this estimation accuracy is represented as a real-valued function $\Phi_\C$ called a utility function, 
the above optimization problem can be formulated as the problem of maximizing $\Phi_\C$ under the LDP constraint.
Such optimization problems have been extensively studied within the framework of classical probability theory 
\cite{Kairouz,Holohan1, Geng1,Geng2,Geng3,Geng4}.
For instance, Kairouz et al.\ \cite{Kairouz} proved that 
the above optimization problem can be reduced to a linear programming problem 
when $\Phi_\C$ can be expressed as a sum of values of a sublinear function.
Except in special cases, further simplification is not known, and 
it is difficult to obtain a general analytical solution to the above optimization problem.

Furthermore, DP was introduced into quantum information theory, and 
quantum extensions of the aforementioned optimization problem have been studied \cite{ISIT2020,Yoshida,Nam, Guan,Nuradha}.
Here, we consider the case of utilizing quantum states $\rho_x$, $x\in\cX$, instead of the classical randomized data $Y$.
Specifically, the data providers convert their private data $X=x$ into quantum states $\rho_x$, and then 
the data collector estimates certain information about $X$ by measuring these quantum states.
In this case, the counterpart of LDP is a condition for the collection of quantum states $(\rho_x)_{x\in\cX}$.
Quantum extensions of LDP, including the one considered here, are increasingly referred to as quantum local differential privacy (QLDP).
The quantum extension considered here mathematically encompasses the classical setting, 
since for each $x\in\cX$ the probability distribution $\mathbb{P}_{Y|X}(\cdot|x)$ can be identified with 
the diagonal density matrix whose diagonal entries are given by $\mathbb{P}_{Y|X}(\cdot|x)$.
Consequently, the classical optimal value never exceeds the quantum optimal value in this framework.
This naturally raises the question whether the quantum optimal value strictly exceeds its classical counterpart, i.e., whether there is a quantum advantage.

Against this background, Yoshida and Hayashi \cite{ISIT2020} proved that there is no quantum advantage when the private data are binary, i.e., $|\cX|=2$.
On the other hand, when $|\cX|\ge3$, Yoshida \cite{Yoshida} derived the classical and quantum optimal values for a certain utility function 
constructed from the right logarithmic derivative (RLD) Fisher information of a one-parameter family, 
thereby showing a quantum advantage for this utility function.
Comparing these optimal values, he evaluated the privacy loss incurred in representing a QLDP mechanism as 
the QLDP mechanism constructed from a classical LDP mechanism and a quantum channel.
However, although his utility function satisfies the data processing inequality, 
it was not motivated by a concrete information-processing scenario.
Motivated by this gap, when $|\cX|\in\{3,4,\dots,9\}$ (resp.\ $|\cX|\in\{4,9\}$), 
Nam et al.\ \cite{Nam} adopted the error exponent in symmetric (resp.\ asymmetric) hypothesis testing as a utility function, and 
proved that there is a quantum advantage in the high-privacy regime (i.e., as $\epsilon\to+0$).
Nevertheless, the following questions remained open: 
\begin{itemize}
	\item
	whether there is a quantum advantage in a concrete information-processing scenario when $|\cX|\ge10$,
	\item
	whether the results of Nam et al.\ \cite{Nam} are optimal, and 
	\item
	whether there exists a QLDP mechanism achieving the quantum optimal value for the utility function of Yoshida \cite{Yoshida}, and 
	if so, what such a mechanism is.
\end{itemize}
This work resolves these questions.

\subsection{Main contributions}
In the high-privacy regime, 
we optimize the error exponents in symmetric and asymmetric hypothesis testing under the LDP and QLDP constraints, and 
provide optimal LDP and QLDP mechanisms.
Consequently, we show quantum advantages for both error exponents in this regime when $|\cX|\ge3$.
Notably, when $|\cX|\ge10$, this work is the first to reveal quantum advantages in concrete information-processing scenarios.
Furthermore, we resolve the optimality question for the parameter ranges covered by Nam et al.\ \cite{Nam}, namely, 
$|\cX|\in\{3,4,\dots,9\}$ in the symmetric case and $|\cX|\in\{4,9\}$ in the asymmetric case.
Specifically, our quantum optimal value strictly improves upon their value 
when $|\cX|\in\{5,6,\dots,9\}$ in the symmetric case and when $|\cX|=9$ in the asymmetric case, 
while coinciding with their value when $|\cX|\in\{3,4\}$ and when $|\cX|=4$, respectively.

In addition to resolving the first and second questions as stated above, 
we analytically investigate the ranges of the parameter $\epsilon$ for which quantum advantages arise in symmetric and asymmetric hypothesis testing.
Specifically, comparing the values of both error exponents for our QLDP mechanisms with the classical optimal values 
naturally leads us to a conjecture concerning equi-isoclinic tight fusion frames (EITFFs), 
which is purely mathematical in nature.
Assuming this conjecture, we show that our QLDP mechanisms also exhibit quantum advantages in a moderate-privacy regime.

Going beyond the two error exponents mentioned above, we 
\begin{itemize}
	\item
	show both the classical and quantum optimal values and 
	\item
	provide both optimal LDP and QLDP mechanisms
\end{itemize}
in the high-privacy regime, for a broader class of utility functions satisfying certain conditions.
Specifically, we assume that a classical utility function $\Phi_\C$ can be expressed as a sum of values of a symmetric sublinear function, 
which corresponds to a symmetric variant of the assumption of Kairouz et al.\ \cite{Kairouz}.
This class of classical utility functions contains many information-theoretic quantities, such as mutual information and the sum of pairwise $f$-divergences.
To show the existence of quantum advantages, 
the above utility function $\Phi_\C$ must be extended to the quantum setting.
To ensure the flexibility and broad applicability of our results, 
rather than assuming a specific form for a quantum extension $\Phi_\Q$ of $\Phi_\C$, 
we assume a specific form for the quadratic term in the Taylor expansion of $\Phi_\Q$.
This class of quantum utility functions contains many information-theoretic quantities, such as Holevo information and the sum of pairwise Petz $f$-divergences.
For such utility functions $\Phi_\Q$ and $\Phi_\C$, 
the asymptotic ratio of the classical and quantum optimal values in the high-privacy regime takes the same value regardless of the utility function.
This asymptotic ratio is greater than or equal to $3/2$ when $|\cX|\ge3$.
One of our main theorems provides a unified explanation for this phenomenon.

While the contributions described above concern the high-privacy regime, 
one of our QLDP mechanisms always achieves the quantum optimal value for the utility function of Yoshida \cite{Yoshida} 
beyond the high-privacy regime.
This provides a clear answer to the third question.

\subsection{Structure of this paper}
The remainder of this paper is organized as follows.
In Section~\ref{pre}, we introduce the necessary notation and definitions, and formulate our optimization problems.
In Section~\ref{main-res}, we explain the second-order Taylor expansions in terms of Fr\'echet derivatives through concrete examples, and 
state our results in the high-privacy regime, including the classical and quantum optimal values and optimal LDP and QLDP mechanisms for concrete utility functions.
Specifically, as utility functions, we investigate the Holevo information and the error exponents in symmetric and asymmetric hypothesis testing.
Beyond the high-privacy regime, we also provide an optimal QLDP mechanism achieving the quantum optimal value for the utility function of Yoshida \cite{Yoshida}.
Section~\ref{exact} states several exact values in symmetric and asymmetric hypothesis testing, 
which are used in Section~\ref{origin}.
Section~\ref{origin} discusses the structural origin of quantum advantage, and 
analytically investigates the ranges of the parameter $\epsilon$ for which quantum advantages arise in symmetric and asymmetric hypothesis testing.
Section~\ref{conc} concludes the paper and discusses directions for future work.
Finally, we prove our results in the Appendices.

\section{Preliminaries}\label{pre}
We first introduce the necessary terms and notation below.
Throughout this paper, $n$ and $d$ denote integers $\ge2$, representing cardinality and dimension, respectively.
For a positive integer $k$, denote by $[k]$ the set $\{1,2,\ldots,k\}$. 
Unless otherwise noted, $r$ and $\epsilon$ denote an integer in $[d-1]$ and a positive number, respectively, 
representing rank and privacy parameter.
Also, $\mathbbm{1}_d$ and $I_d$ denote the all-ones column vector in $\mathbb{R}^d$ and the identity matrix on $\mathbb{C}^d$, respectively.
A quantum state is a density matrix on $\mathbb{C}^d$ and is often denoted by $\rho$.
A classical state is a probability column vector in $\mathbb{R}^d$ and is often denoted by $p$.
Denote by $\diag(z_1,\ldots,z_d)$ the diagonal matrix whose diagonal entries are $z_1,\ldots,z_d$.
A classical state $p$ can be regarded as the quantum state $\diag(p)=\diag(p(1),\ldots,p(d))$.
Since we often use $n$ quantum states $\rho_1,\ldots,\rho_n$, 
they are collectively denoted by $\vec{\rho}$.
Also, we often use a column-stochastic matrix (or conditional probability distribution) $q=[q(y|x)]_{y\in[d],x\in[n]}$ 
to denote the $n$ classical states $q(\cdot|1),\ldots,q(\cdot|n)$.
For a column-stochastic matrix $q\in\mathbb{R}^{d\times n}$, 
we define $\vec{\Diag}(q)=(\Diag_1(q),\ldots,\Diag_n(q))$ as $\Diag_x(q)=\diag(q(\cdot|x))$.
Denote by $\|\cdot\|$ and $\|\cdot\|_1$ the operator norm and the trace norm, respectively.
For a real number $t$, denote by $\lfloor{t}\rfloor$ (resp.\ $\lceil{t}\rceil$) the greatest (resp.\ least) integer $\le t$ (resp.\ $\ge t$).

Furthermore, we need to introduce another term, namely, measurement.
A measurement is a positive operator-valued measure (POVM) and is often denoted by $(M_y)_{y\in\cY}$, 
where $\cY$ is a finite set.
That is, $(M_y)_{y\in\cY}$ is called a measurement if all $M_y$ are positive semi-definite matrices on $\mathbb{C}^d$ satisfying $\sum_{y\in\cY} M_y=I_d$.
When a quantum state $\rho$ is measured by a measurement $(M_y)_{y\in\cY}$, 
each outcome $y$ is observed with probability $\Tr\rho M_y$.
Hence, $(\Tr\rho M_y)_{y\in\cY}$ is a probability vector.

\vskip1ex
\textbf{Local differential privacy.}---%
As stated in Section~\ref{intro}, LDP is a condition for conditional probability distributions.
For a positive number $\epsilon$ and a conditional probability distribution $q\in\mathbb{R}^{d\times n}$, 
we say that $q$ is \textit{$\epsilon$-LDP (or an $\epsilon$-LDP mechanism)} if 
\[
q(y|x') \le e^\epsilon q(y|x)
\]
for all $x,x'\in[n]$ and $y\in[d]$ \cite{Duchi}.
Here, the parameter $\epsilon$ represents the privacy loss bound; 
a smaller (resp.\ larger) $\epsilon$ corresponds to a higher (resp.\ lower) privacy level.
If $q$ is an $\epsilon$-LDP mechanism, 
then the probability vectors $q(\cdot|1),\ldots,q(\cdot|n)$ have the same support.
Hence, we often implicitly assume that $q(y|x)>0$ for all $x\in[n]$ and $y\in[d]$.

\vskip1ex
\textbf{Quantum local differential privacy.}---%
Once regarding LDP as a condition for $n$ probability vectors, 
it is naturally extended as follows.
For a positive number $\epsilon$ and $n$ quantum states $\rho_1,\ldots,\rho_n$ on $\mathbb{C}^d$, 
we say that $\vec{\rho}=(\rho_1,\ldots,\rho_n)$ is \textit{$\epsilon$-QLDP (or an $\epsilon$-QLDP mechanism)} if 
\begin{equation}
	\rho_{x'} \le e^\epsilon\rho_x \label{eq01}
\end{equation}
for all $x,x'\in[n]$ \cite{ISIT2020}, where the inequality means for the difference $e^\epsilon\rho_x-\rho_{x'}$ to be positive semi-definite.
Instead of the above definition, QLDP can also be defined by using LDP and measurements.
Indeed, the following conditions are equivalent to each other \cite{ISIT2020}: 
\begin{itemize}
	\item
	$\vec{\rho}$ is an $\epsilon$-QLDP mechanism;
	\item
	$q$ defined as $q(y|x)=\Tr\rho_xM_y$ is an $\epsilon$-LDP mechanism for every measurement $(M_y)_{y\in\cY}$.
\end{itemize}
If $\vec{\rho}=(\rho_1,\ldots,\rho_n)$ is an $\epsilon$-QLDP mechanism, 
then the density matrices $\rho_1,\ldots,\rho_n$ have the same support, i.e., 
the ranges of $\rho_1,\ldots,\rho_n$ are equal to one another.
Hence, we often implicitly assume that $\rho_x$ has full rank for every $x\in[n]$.

\vskip1ex
\textbf{Formulation of optimization problems.}---%
As stated in Section~\ref{intro}, we optimize the trade-off between privacy and utility in the following formulation.
Let $\Phi_\C$ be a real-valued function of an $\epsilon$-LDP mechanism, and 
$\Phi_\Q$ be a quantum extension of $\Phi_\C$.
Here, we say that $\Phi_\Q$ is a \textit{quantum extension} of $\Phi_\C$ if 
\begin{itemize}
	\item
	$\Phi_\Q$ is a real-valued function of an $\epsilon$-QLDP mechanism; 
	\item
	$\Phi_\Q(U\rho_1U^\dag,\ldots,U\rho_nU^\dag)=\Phi_\Q(\vec{\rho})$ 
	for every unitary matrix $U$ and every $\epsilon$-QLDP mechanism $\vec{\rho}=(\rho_1,\ldots,\rho_n)$; 
	\item
	$\Phi_\Q(\vec{\Diag}(q))=\Phi_\C(q)$ for every $\epsilon$-LDP mechanism.
\end{itemize}
The second condition is referred to as the unitary invariance.
Then we consider the following classical and quantum optimization problems: 
\begin{equation*}
	\OPT_n(\epsilon; \Phi_\C) = \underbrace{\sup_{q:\,\text{$\epsilon$-LDP}}}_{\text{Privacy protection}} \underbrace{\Phi_\C(q)}_{\text{Utility}}
	\quad\text{and}\quad
	\OPT_n(\epsilon; \Phi_\Q) = \underbrace{\sup_{\vec{\rho}:\,\text{$\epsilon$-QLDP}}}_{\text{Privacy protection}} \underbrace{\Phi_\Q(\vec{\rho})}_{\text{Utility}}.
\end{equation*}
Note that neither the dimension of the underlying space $\mathbb{R}^d$ nor $\mathbb{C}^d$ is fixed in the above optimizations.
Optimization problems of this type, including those above, are often considered in information-theoretic studies of LDP and QLDP 
\cite{Kairouz,Holohan1, Geng1,Geng2,Geng3,Geng4, ISIT2020,Yoshida,Nam, Guan,Nuradha}.
Since $\Phi_\Q$ is a quantum extension of $\Phi_\C$, 
the quantum optimal value $\OPT_n(\epsilon; \Phi_\Q)$ is greater than or equal to its classical counterpart $\OPT_n(\epsilon; \Phi_\C)$.
Therefore, our primary interest is whether the quantum optimal value strictly exceeds its classical counterpart.

Furthermore, we say that $\Phi_\Q$ satisfies the \textit{data processing inequality} if 
\begin{equation}
	\Phi_\Q(\Lambda(\rho_1),\ldots,\Lambda(\rho_n)) \le \Phi_\Q(\vec{\rho})
	\label{eq-dpi-q}
\end{equation}
for every $\epsilon$-QLDP mechanism $\vec{\rho}=(\rho_1,\ldots,\rho_n)$ and every quantum channel (i.e., completely positive and trace-preserving map) $\Lambda$.
Its classical counterpart is below: 
\begin{equation}
	\Phi_\C(\Lambda q) \le \Phi_\C(q)
	\label{eq-dpi-c}
\end{equation}
for every $\epsilon$-LDP mechanism $q$ and every classical channel (i.e., column-stochastic matrix) $\Lambda$.
Inequality \eqref{eq-dpi-c} follows from \eqref{eq-dpi-q}, since $\Phi_\Q$ is a quantum extension of $\Phi_\C$.
An information-theoretic quantity usually satisfies the data processing inequality.
Almost all utility functions we investigate in this paper actually satisfy the data processing inequality.

As stated in Section~\ref{intro}, we assume that $\Phi_\C$ is of the form 
\begin{equation}
	\Phi_\C(q) = \sum_{\substack{y\in[d] \\ q(y|1),\ldots,q(y|n)>0}} \phi(q(y|1),\ldots,q(y|n)),
	\label{eq-sum-subl}
\end{equation}
where $q$ is an $\epsilon$-LDP mechanism in $\mathbb{R}^{d\times n}$, and 
$\phi\colon (0,\infty)^n\to\mathbb{R}$ is a sublinear function, i.e., satisfies the following conditions: 
\begin{itemize}
	\item
	$\phi(\alpha\mathbf{z})=\alpha\phi(\mathbf{z})$ for every $\alpha>0$ and every $\mathbf{z}\in(0,\infty)^n$;
	\item
	$\phi(\mathbf{z}+\mathbf{z}')\le\phi(\mathbf{z})+\phi(\mathbf{z}')$ for all $\mathbf{z},\mathbf{z}'\in(0,\infty)^n$.
\end{itemize}
Then $\Phi_\C$ satisfies the data processing inequality; see also \cite[Proposition~17]{Kairouz}.
Many information-theoretic quantities, e.g., mutual information and $f$-divergences, can be written as \eqref{eq-sum-subl}.

To understand the above formulation, we show the following lemma, which is implicitly used in this paper.

\begin{lemma}\label{lemFormulation}
	Let $\rho_0$ be a quantum state.
	Then $\Phi_\Q(\rho_0,\ldots,\rho_0)=\phi(\mathbbm{1}_n)$.
\end{lemma}
\begin{proof}
	It is clear that the collection of the $n$ quantum states $\rho_0,\ldots,\rho_0$ is $\epsilon$-QLDP.
	Diagonalize $\rho_0$ by a unitary matrix $U$: 
	$U\rho_0U^\dag=\diag(p_0)$ for some probability vector $p_0$.
	The unitary invariance of $\Phi_\Q$ implies 
	\[
	\Phi_\Q(\rho_0,\ldots,\rho_0)
	= \Phi_\Q(\diag(p_0),\ldots,\diag(p_0))
	= \Phi_\C(p_0,\ldots,p_0).
	\]
	By this and \eqref{eq-sum-subl}, 
	\[
	\Phi_\C(p_0,\ldots,p_0)
	= \sum_{y:\,p_0(y)>0} \phi(p_0(y),\ldots,p_0(y))
	= \sum_{y:\,p_0(y)>0} p_0(y)\phi(\mathbbm{1}_n)
	= \phi(\mathbbm{1}_n).
	\]
	Therefore, the desired equality holds.
\end{proof}

\section{Main results}\label{main-res}
As before, let $d$ and $n$ be integers $\ge2$, $\epsilon$ be a positive number, and 
assume that $\Phi_\C$ is of the form \eqref{eq-sum-subl}, and $\Phi_\Q$ is its quantum extension.
This section presets 
\begin{itemize}
	\item[(a)]
	the classical and quantum optimal values $\OPT_n(\epsilon; \Phi_\C)$ and $\OPT_n(\epsilon; \Phi_\Q)$, 
	\item[(b)]
	optimal $\epsilon$-QLDP mechanisms achieving $\OPT_n(\epsilon; \Phi_\Q)$, and 
	\item[(c)]
	optimal $\epsilon$-LDP mechanisms achieving $\OPT_n(\epsilon; \Phi_\C)$
\end{itemize}
in the high-privacy regime, i.e., as $\epsilon\to+0$.
Theorem~\ref{main1}, which is stated later, asserts 
\begin{equation}
	\lim_{\epsilon\to+0} \frac{\OPT_n(\epsilon; \Phi_\Q)}{\OPT_n(\epsilon; \Phi_\C)}
	= \frac{n(n-1)}{2\lfloor{n/2}\rfloor\lceil{n/2}\rceil}
	\ge \frac{3}{2} \label{eq02}
\end{equation}
under the assumptions of the theorem, including $\phi(\mathbbm{1}_n)=0$ and $n\ge3$.
This establishes quantum advantages for many information-theoretic quantities, 
including the Holevo information and the error exponents in symmetric and asymmetric hypothesis testing.
For (b), we propose \textit{isoclinic mechanisms}, 
which are a significant extension of the mechanisms proposed by Nam et al.\ \cite{Nam}.
For (c), we employ binary mechanisms \cite{Kairouz}.

\subsection{Classical and quantum optimal values}\label{OptVals}
This subsection presents the classical and quantum optimal values $\OPT_n(\epsilon; \Phi_\C)$ and $\OPT_n(\epsilon; \Phi_\Q)$.
To this end, we explain the necessary notions below.

\vskip1ex
\textbf{Petz monotone metrics.}---%
Let $f\colon (0,\infty)\to(0,\infty)$ be an operator monotone function satisfying $f(1)>0$ and $f(t)=tf(t^{-1})$.
For a full-rank density matrix $\rho_0$ on $\mathbb{C}^d$, 
define the Petz monotone metric $J_{\rho_0}^f$ as \cite{Petz1,Petz2,Petz3} 
\begin{equation}
	J_{\rho_0}^f[X,Y] = \sum_{i,j\in[d]} \frac{1}{\lambda_jf(\lambda_i/\lambda_j)}\braket{j|X|i}\braket{i|Y|j},
	\label{eq03}
\end{equation}
where $\rho_0=\sum_{i\in[d]} \lambda_i\ketbra{i}{i}$ is the diagonalization of $\rho_0$ by a unitary matrix.
Then $J_{\rho_0}^f[X,X]$ satisfies the data processing inequality \cite[Theorem~4.1]{Petz2} and 
\begin{equation}
	J_{\rho_0}^{\SLD}[X,X] \le f(1)J_{\rho_0}^f[X,X] \le J_{\rho_0}^{\RLD}[X,X],
	\label{eq-SLD-RLD}
\end{equation}
where $J_{\rho_0}^{\SLD}$ is the symmetric logarithmic derivative (SLD) Fisher metric, and 
$J_{\rho_0}^{\RLD}$ is the right logarithmic derivative (RLD) Fisher metric \cite[Eq.~(4.12)]{Petz2}, \cite[p.~256]{book1}.
If the normalization condition $f(1)=1$ holds, 
then $J_{\rho_0}^f$ is called the normalized Petz monotone metric.
The normalization ensures that the metric coincides with the classical Fisher information metric for commutative matrices.

For example, the following normalized metrics are famous and often appear in quantum information theory.
\begin{itemize}
	\item
	If $f(t)=(1+t)/2$, then $J_{\rho_0}^f$ is called the SLD Fisher metric, 
	which is denoted by $J_{\rho_0}^{\SLD}$: 
	\begin{align*}
		J_{\rho_0}^{\mathrm{SLD}}[X,Y]
		&= \sum_{i,j\in[d]} \frac{2}{\lambda_i+\lambda_j}\braket{j|X|i}\braket{i|Y|j}\\
		&= \int_0^\infty \Tr e^{-(t/2)\rho_0}Xe^{-(t/2)\rho_0}Y\,\mathrm{d}t.
	\end{align*}
	\item
	If $f(t)=2t/(1+t)$, then $J_{\rho_0}^f$ is called the RLD Fisher metric, 
	which is denoted by $J_{\rho_0}^{\RLD}$: 
	\begin{align*}
		J_{\rho_0}^{\mathrm{RLD}}[X,Y]
		&= \sum_{i,j\in[d]} \frac{\lambda_i^{-1}+\lambda_j^{-1}}{2}\braket{j|X|i}\braket{i|Y|j}\\
		&= \frac{1}{2}\Tr(XY+YX)\rho_0^{-1}.
	\end{align*}
	\item
	If $f(t)=(t-1)/\ln t$, then $J_{\rho_0}^f$ is called the Bogoliubov--Kubo--Mori (BKM) metric, 
	which is denoted by $J_{\rho_0}^{\BKM}$: 
	\begin{align*}
		J_{\rho_0}^{\BKM}[X,Y]
		&= \sum_{i,j\in[d]} \frac{\ln\lambda_i-\ln\lambda_j}{\lambda_i-\lambda_j}\braket{j|X|i}\braket{i|Y|j}\\
		&= \int_0^\infty \Tr X(\rho_0+tI_d)^{-1}Y(\rho_0+tI_d)^{-1}\,\mathrm{d}t.
	\end{align*}
	\item
	If $f(t)=s(1-s)(t-1)^2(t^s-1)^{-1}(t^{1-s}-1)^{-1}$ and $0<s<1$, then $J_{\rho_0}^f$ is called the Wigner--Yanase--Dyson (WYD) metric, 
	which is denoted by $J_{\rho_0}^{\WYD(s)}$: 
	\begin{align*}
		J_{\rho_0}^{\WYD(s)}[X,Y]
		&= \frac{1}{s(1-s)}\sum_{i,j\in[d]} \frac{(\lambda_i^s-\lambda_j^s)(\lambda_i^{1-s}-\lambda_j^{1-s})}{(\lambda_i-\lambda_j)^2}\braket{j|X|i}\braket{i|Y|j}.
	\end{align*}
	In particular, the case $s=1/2$ is below: 
	\begin{align*}
		J_{\rho_0}^{\WYD(1/2)}[X,Y]
		&= \sum_{i,j\in[d]} \frac{4}{(\sqrt{\lambda_i}+\sqrt{\lambda_j})^2}\braket{j|X|i}\braket{i|Y|j}.
	\end{align*}
\end{itemize}

\vskip1ex
\textbf{Taylor expansions in terms of Fr\'echet derivatives.}---%
In the theorem stated later, we assume that $\Phi_\Q(\vec{\rho})$ has the second-order Taylor expansion around $(\rho_{\avg},\ldots,\rho_{\avg})$, 
where $\rho_{\avg}$ is defined as $\rho_{\avg}=(1/n)\sum_{x\in[n]} \rho_x$.
Hence, we explain Taylor expansions of a few functions in terms of Fr\'echet derivatives.

Let $\rho_0$ be a full-rank quantum state on $\mathbb{C}^d$, and set $\rho_x=\rho_0+\delta\rho_x$ for $x\in[n]$.
We say that $\Phi_\Q(\vec{\rho})$ has the second-order Taylor expansion around $(\rho_0,\ldots,\rho_0)$ if 
\begin{gather*}
	\Phi_\Q(\vec{\rho}) = \phi(\mathbbm{1}_n) + \Phi_\Q^{(1)}(\rho_0; \delta\vec{\rho}) + \Phi_\Q^{(2)}(\rho_0; \delta\vec{\rho})
	+ o\Bigl( \max_{i\in[n]} \|\delta\rho_i\|^2 \Bigr),\\
	\Phi_\Q^{(1)}(\rho_0; \delta\vec{\rho}) = \sum_{i\in[n]} \cA_{i;\rho_0}[\delta\rho_i], \quad\text{and}\quad
	\Phi_\Q^{(2)}(\rho_0; \delta\vec{\rho}) = \frac{1}{2}\sum_{i,j\in[n]} \cB_{i,j;\rho_0}[\delta\rho_i,\delta\rho_j]
\end{gather*}
as $\max_{i\in[n]} \|\delta\rho_i\|\to0$, where $\cA_{i;\rho_0}[X]$, $i\in[n]$, are linear in $X$, and 
$\cB_{i,j;\rho_0}[X,Y]$, $i,j\in[n]$, are bilinear in $X$ and $Y$.
If $\Phi_\Q(\vec{\rho})$ has the second-order Taylor expansion around $(\rho_0,\ldots,\rho_0)$, 
then $\Phi_\Q(\vec{\rho})$ is Fr\'echet differentiable with respect to each variable at $(\rho_0,\ldots,\rho_0)$,\footnote{
Throughout this paper, Fr\'echet differentiability is understood with respect to variations in the space of traceless Hermitian matrices.} 
and the linear map $\cA_{i;\rho_0}$ is equal to the first-order partial Fr\'echet derivative of $\Phi_\Q(\vec{\rho})$ 
at $(\rho_0,\ldots,\rho_0)$ with respect to the $i$th variable.
However, the existence of the above second-order Taylor expansion does not necessarily imply 
the existence of the second-order partial Fr\'echet derivatives at $(\rho_0,\ldots,\rho_0)$.
These facts are analogous to those in multivariable calculus.

For example, the second-order Taylor expansions of Petz $F$-divergences and trace functions are as follows.
\begin{itemize}
	\item
	Let $F\colon (0,\infty)\to\mathbb{R}$ be an operator convex function satisfying $F(1)=0$.
	The Petz $F$-divergence $D_F(\rho_1\|\rho_2)$ is defined as \cite[Definition~2.1]{Lesniewski}, \cite[p.~290]{book1} 
	\begin{equation}
		D_F(\rho_1\|\rho_2) = \sum_{i,j\in[d]} \lambda_{2,j}F(\lambda_{1,i}/\lambda_{2,j})\Tr E_{1,i}E_{2,j},
		\label{eq04}
	\end{equation}
	where $\rho_k=\sum_{i\in[d]} \lambda_{k,i}E_{k,i}$ is the spectral decomposition of $\rho_k$.
	For classical states $p_1$ and $p_2$, 
	the $F$-divergence $D_F(p_1\|p_2)$ is defined as $D_F(\diag(p_1)\|\diag(p_2))$.
	The Petz $F$-divergence $D_F(\rho_1\|\rho_2)$ satisfies the data processing inequality \cite[Theorem~2.4]{Lesniewski}, and 
	has the second-order Taylor expansion around $(\rho_0,\rho_0)$ (essentially \cite[Theorems~2.8--2.9]{Lesniewski}; see Appendix~\ref{appTaylorFdiv}): 
	\begin{equation}
		D_F(\rho_1\|\rho_2) = \frac{1}{2}J_{\rho_0}^f[\delta\rho_1-\delta\rho_2,\delta\rho_1-\delta\rho_2]
		+ o(\max\{ \|\delta\rho_1\|^2,\|\delta\rho_2\|^2 \}),
		\label{eqTaylorFdiv}
	\end{equation}
	where $J_{\rho_0}^f$ is the Petz monotone metric defined in \eqref{eq03} with the following function $f$: 
	\begin{equation}
		f(t) = \frac{(t-1)^2}{F(t)+tF(t^{-1})}.
		\label{eq05}
	\end{equation}
	In particular, the case $F(t)=t\ln t$ is below (cf.\ \cite[Theorem~3]{Grace}): 
	\begin{equation*}
		D(\rho_1\|\rho_2) = \frac{1}{2}J_{\rho_0}^{\mathrm{BKM}}[\delta\rho_1-\delta\rho_2,\delta\rho_1-\delta\rho_2]
		+ o(\max\{ \|\delta\rho_1\|^2,\|\delta\rho_2\|^2 \}),
	\end{equation*}
	where $D(\rho_1\|\rho_2)$ denotes the Umegaki relative entropy of $\rho_1$ with respect to $\rho_2$.
	\item
	Let $f\colon (0,\infty)\to\mathbb{R}$ be of class $C^2$.
	The trace function $\Tr f(\rho)$ has the second-order Taylor expansion around $\rho_0$ 
	(cf.\ \cite[Theorems~3.23, 3.33, and the proof of Theorem~3.27]{book3}, \cite[Theorem~1]{Grace}): 
	\[
	\Tr f(\rho) = \Tr f(\rho_0) + \Tr f'(\rho_0)\delta\rho
	+ \frac{1}{2}\sum_{i,j\in[n]} (f')^{[1]}(\lambda_i,\lambda_j)\abs{\braket{i|\delta\rho|j}}^2
	+ o(\|\delta\rho\|^2),
	\]
	where $\rho_0=\sum_{i\in[d]} \lambda_i\ketbra{i}{i}$ is the diagonalization of $\rho_0$ by a unitary matrix, and 
	$(f')^{[1]}$ is the divided difference of $f'$ defined as 
	\[
	(f')^{[1]}(t_1,t_2) = \frac{f'(t_1)-f'(t_2)}{t_1-t_2}
	\]
	if $t_1\not=t_2$, and $(f')^{[1]}(t_1,t_2)=f''(t_1)$ if $t_1=t_2$.
	In particular, the case $f(t)=t\ln t$ is below (cf.\ \cite[Theorem~2]{Grace}): 
	\begin{equation}
		-H(\rho) = -H(\rho_0) + \Tr(\ln\rho_0)\delta\rho + \frac{1}{2}J_{\rho_0}^{\mathrm{BKM}}[\delta\rho,\delta\rho] + o(\|\delta\rho\|^2),
		\label{eq07}
	\end{equation}
	where $H(\rho)$ denotes the von Neumann entropy of $\rho$.
\end{itemize}

\vskip1ex
\textbf{Main results (1).}---%
We now present the classical and quantum optimal values $\OPT_n(\epsilon; \Phi_\C)$ and $\OPT_n(\epsilon; \Phi_\Q)$.
Recall that $\rho_{\avg}$ is defined as $\rho_{\avg}=(1/n)\sum_{x\in[n]} \rho_x$.

\begin{theorem}\label{main1}
	Assume that 
	\renewcommand{\labelenumi}{$\arabic{enumi}$.}
	\begin{enumerate}
		\item
		$\phi$ is of class $C^2$ in a neighborhood of $\mathbbm{1}_n$ and symmetric;
		\item
		$\beta_0 \coloneqq \partial_1^2\phi(\mathbbm{1}_n)>0$;
		\item
		$\Phi_\Q$ is symmetric;
		\item
		the quadratic term in the Taylor expansion of $\Phi_\Q(\vec{\rho})$ around $(\rho_{\avg},\ldots,\rho_{\avg})$ is 
		\begin{equation}
			\Phi_\Q^{(2)}(\rho_{\avg}; \delta\vec{\rho})
			= \frac{\beta_0n}{2(n-1)}\sum_{i\in[n]} J_{\rho_{\avg}}[\delta\rho_i,\delta\rho_i]
			\label{eq08}
		\end{equation}
		for some normalized Petz monotone metric $J_{\rho_{\avg}}$.
	\end{enumerate}
	Then 
	\begin{align*}
		\OPT_n(\epsilon; \Phi_\C) &= \phi(\mathbbm{1}_n) + \frac{\lfloor{n/2}\rfloor\lceil{n/2}\rceil}{n-1}\cdot\frac{\beta_0}{2}\epsilon^2 + o(\epsilon^2)
		\quad\text{and}\quad\\ 
		\OPT_n(\epsilon; \Phi_\Q) &= \phi(\mathbbm{1}_n) + \frac{\beta_0n}{4}\epsilon^2 + o(\epsilon^2)
	\end{align*}
	as $\epsilon\to+0$.
	If the additional assumptions $\phi(\mathbbm{1}_n)=0$ and $n\ge3$ hold, 
	then \eqref{eq02} holds.
\end{theorem}

Actually, \eqref{eq08} can also be written as 
\[
\Phi_\Q^{(2)}(\rho_{\avg}; \delta\vec{\rho})
= \frac{1}{2}\sum_{i,j\in[n]} \partial_i\partial_j\phi(\mathbbm{1}_n)J_{\rho_{\avg}}[\delta\rho_i,\delta\rho_j]
= \frac{1}{2}\Tr \mathbf{H}_\phi(\mathbbm{1}_n)\mathbf{J}_{\rho_{\avg}}(\delta\vec{\rho}),
\]
where $\mathbf{H}_\phi$ is the Hesse matrix of $\phi$, and 
$\mathbf{J}_{\rho_{\avg}}(\delta\vec{\rho})$ is the $n\times n$ matrix whose $(i,j)$-entry is $J_{\rho_{\avg}}[\delta\rho_i,\delta\rho_j]$.
While this is more natural from a theoretical perspective, 
\eqref{eq08} is more practical to verify.

\vskip1ex
\textbf{Holevo information.}---%
For example, assume that $\Phi_\Q$ is the Holevo information 
\begin{equation}
	\chi(p_{\mix}; \vec{\rho}) = H(\rho_{\avg}) - \sum_{x\in[n]} p_{\mix}(x)H(\rho_x),
	\label{eq09}
\end{equation}
where $p_{\mix}$ denotes the uniform distribution on $[n]$.
If $\vec{\rho}=\vec{\Diag}(q)$, 
then $\chi(p_{\mix}; \vec{\rho})$ is the mutual information of the joint distribution $p_{\mix}(x)q(y|x)$, 
which is denoted by $\chi(p_{\mix}; q)$.
Setting $L(t)=t\ln t$, we can write the function $\phi$ as 
\[
\phi(\mathbf{z}) = -L\biggl( \frac{1}{n}\sum_{i\in[n]} z_i \biggr) + \frac{1}{n}\sum_{i\in[n]} L(z_i).
\]
The function $\phi$ is sublinear, symmetric, and of class $C^2$.
Also, $\partial_1^2\phi(\mathbbm{1}_n)=n^{-2}(n-1)>0$.
By \eqref{eq07}, the Holevo information $\chi(p; \vec{\rho})$ has the second-order Taylor expansion around $(\rho_{\avg},\ldots,\rho_{\avg})$: 
\begin{align*}
	\chi(p_{\mix}; \vec{\rho})
	&= \frac{1}{2n}\sum_{i\in[n]} J_{\rho_{\avg}}^{\mathrm{BKM}}[\delta\rho_i,\delta\rho_i] + o\Bigl( \max_{i\in[n]} \|\delta\rho_i\|^2 \Bigr)\\
	&= \frac{\beta_0n}{2(n-1)}\sum_{i\in[n]} J_{\rho_{\avg}}^{\mathrm{BKM}}[\delta\rho_i,\delta\rho_i] + o\Bigl( \max_{i\in[n]} \|\delta\rho_i\|^2 \Bigr),
\end{align*}
where $\beta_0=\partial_1^2\phi(\mathbbm{1}_n)=n^{-2}(n-1)$.
Since the other assumptions clearly hold, 
Theorem~\ref{main1} implies the following corollary.

\begin{corollary}\label{coro1}
	We have 
	\begin{align*}
		\sup_{q:\,\text{$\epsilon$-LDP}} \chi(p_{\mix}; q)
		&= \frac{\lfloor{n/2}\rfloor\lceil{n/2}\rceil}{2n^2}\epsilon^2 + o(\epsilon^2)
		\quad\text{and}\quad\\
		\sup_{\vec{\rho}:\,\text{$\epsilon$-QLDP}} \chi(p_{\mix}; \vec{\rho})
		&= \frac{n-1}{4n}\epsilon^2 + o(\epsilon^2)
	\end{align*}
	as $\epsilon\to+0$.
\end{corollary}

\vskip1ex
\textbf{Error exponents in symmetric and asymmetric hypothesis testing.}---%
As applications of Theorem~\ref{main1}, we can optimize the error exponents in symmetric and asymmetric hypothesis testing.
Here, we adopt the scenarios considered by Nam et al.\ \cite{Nam} as follows.
For $k\in[n]$ and $\eta\in(0,1]$, set the probability distribution $p_k^\eta$ on $[n]$ as 
\[
p_k^\eta(x) = \eta\delta_{k,x} + \frac{1-\eta}{n}
\quad(x\in[n]),
\]
where $\delta_{i,j}$ is the Kronecker delta.
Recall that $p_{\mix}$ is the uniform distribution on $[n]$.
Assume that the data collector is interested in the probability distribution $\mathbb{P}_X$ of private data $X$.
Then symmetric and asymmetric scenarios are as follows. 
\begin{itemize}
	\item
	\textbf{Symmetric scenario.}
	Assuming that $\mathbb{P}_X$ is one of $p_1^\eta,\ldots,p_n^\eta$, 
	the data collector aims to investigate which of them is the most appropriate.
	\item
	\textbf{Asymmetric scenario.}
	The data collector aims to investigate whether the private data are biased toward a single value.
	For this aim, consider the following null and alternative hypotheses.
	\begin{itemize}
		\item
		Null hypothesis: $p\in\{ p_1^\eta,\ldots,p_n^\eta \}$.
		\item
		Alternative hypothesis: $p=p_{\mix}$.
	\end{itemize}
\end{itemize}
To protect the private data, 
the data providers release other data $Y$ subject to an $\epsilon$-LDP mechanism $\mathbb{P}_{Y|X}$ instead of $X$.
As stated in Section~\ref{intro}, its quantum version can also be considered by using an $\epsilon$-QLDP mechanism $\vec{\rho}=(\rho_1,\ldots,\rho_n)$.

In the symmetric scenario, the minimum error probability converges to zero exponentially as the number of i.i.d.\ samples tends to infinity.
Then the error exponent \cite[Eqs.~(5) and (6)]{Nam} is 
\begin{equation}
	S(\vec{\rho}; p_1^\eta,\ldots,p_n^\eta) = \min_{i\not=j} C(\tilde{\rho}_i,\tilde{\rho}_j),
	\label{eq10}
\end{equation}
where $\tilde{\rho}_k$ is defined as 
\begin{equation}
	\tilde{\rho}_k = \sum_{x\in[n]} p_k^\eta(x)\rho_x
	\label{eq11}
\end{equation}
for $k\in[n]$, and 
$C(A,B)$ is the Chernoff information defined as 
\[
C(A,B) = -\ln\min_{s\in[0,1]} \Tr A^sB^{1-s}
\]
for positive semi-definite matrices $A$ and $B$.
The Chernoff information satisfies the data processing inequality \cite[Eq.~(10)]{Audenaert}.
For the classical case, set $\mathbb{P}_{Y|X}$ as $q=\mathbb{P}_{Y|X}$ and 
define $\tilde{q}$ as $\tilde{q}=q[p_1^\eta,\ldots,p_n^\eta]$.
If replacing $\tilde{\rho}_i$ and $\tilde{\rho}_j$ in \eqref{eq10} with $\Diag_i(\tilde{q})$ and $\Diag_j(\tilde{q})$, respectively, 
we obtain the error exponent in the classical case \cite[Eq.~(8)]{Nam}: 
\begin{equation}
	S(q; p_1^\eta,\ldots,p_n^\eta) = \min_{i\not=j} C(\Diag_i(\tilde{q}),\Diag_j(\tilde{q})).
	\label{eq10'}
\end{equation}

In the asymmetric scenario, two error probabilities are considered: namely, the Type I and Type II error probabilities.
The Type I error probability is that of mistakenly rejecting the null hypothesis as a result of testing.
The Type II error probability is that of failing to reject the null hypothesis as a result of testing.
Usually, one minimizes the Type II error probability when the Type I error probability does not exceed a threshold in the interval $(0,1)$.
The minimum Type II error probability converges to zero exponentially as the number of i.i.d.\ samples tends to infinity.
Then the error exponent \cite[Eqs.~(11) and (12)]{Nam} is 
\begin{equation}
	A(\vec{\rho}; p_1^\eta,\ldots,p_n^\eta; p_{\mix}) = \min_{x\in[n]} D(\tilde{\rho}_x\|\rho_{\avg}),
	\label{eq12}
\end{equation}
where we note $\sum_{x\in[n]} p_{\mix}(x)\rho_x=\rho_{\avg}$.
Set $\Diag_{\avg}(q)$ as $\Diag_{\avg}(q)=(1/n)\sum_{x\in[n]} \Diag_x(q)$.
If replacing $\tilde{\rho}_i$ and $\rho_{\avg}$ in \eqref{eq12} with $\Diag_i(\tilde{q})$ and $\Diag_{\avg}(q)$, respectively, 
we obtain the error exponent in the classical case \cite[Eq.~(14)]{Nam}: 
\[
A(\vec{\rho}; p_1^\eta,\ldots,p_n^\eta; p_{\mix}) = \min_{x\in[n]} D(\Diag_x(\tilde{q})\|\Diag_{\avg}(q)).
\]

\vskip1ex
\textbf{Main results (2).}---%
For simplicity, set $S^\eta(q)$, $S^\eta(\vec{\rho})$, $A^\eta(q)$, and $A^\eta(\vec{\rho})$ as 
\begin{gather*}
	S^\eta(q) = S(q; p_1^\eta,\ldots,p_n^\eta),\quad
	S^\eta(\vec{\rho}) = S(\vec{\rho}; p_1^\eta,\ldots,p_n^\eta),\\
	A^\eta(q) = A(q; p_1^\eta,\ldots,p_n^\eta; p_{\mix}), \quad\text{and}\quad
	A^\eta(\vec{\rho}) = A(\vec{\rho}; p_1^\eta,\ldots,p_n^\eta; p_{\mix}),
\end{gather*}
respectively.
We can optimize them in the high-privacy regime as follows.

\begin{theorem}\label{main1'}
	Let $\eta\in(0,1]$.
	Then 
	\begin{align}
		S_\C^\eta(n,\epsilon) &\coloneqq \sup_{q:\,\text{$\epsilon$-LDP}} S^\eta(q)
		= \frac{\lfloor{n/2}\rfloor\lceil{n/2}\rceil}{4n(n-1)}\eta^2\epsilon^2 + o(\eta^2\epsilon^2),\label{eqSC}\\
		S_\Q^\eta(n,\epsilon) &\coloneqq \sup_{\vec{\rho}:\,\text{$\epsilon$-QLDP}} S^\eta(\vec{\rho})
		= \frac{1}{8}\eta^2\epsilon^2 + o(\eta^2\epsilon^2),\label{eqSQ}\\
		A_\C^\eta(n,\epsilon) &\coloneqq \sup_{q:\,\text{$\epsilon$-LDP}} A^\eta(q)
		= \frac{\lfloor{n/2}\rfloor\lceil{n/2}\rceil}{2n^2}\eta^2\epsilon^2 + o(\eta^2\epsilon^2),\label{eqAC}\\
		A_\Q^\eta(n,\epsilon) &\coloneqq \sup_{\vec{\rho}:\,\text{$\epsilon$-QLDP}} A^\eta(\vec{\rho})\label{eqAQ}
		= \frac{n-1}{4n}\eta^2\epsilon^2 + o(\eta^2\epsilon^2)
	\end{align}
	as $\epsilon\to+0$.
	If $n\ge3$, then 
	\[
	\lim_{\epsilon\to+0} \frac{S_\Q^\eta(n,\epsilon)}{S_\C^\eta(n,\epsilon)}
	= \lim_{\epsilon\to+0} \frac{A_\Q^\eta(n,\epsilon)}{A_\C^\eta(n,\epsilon)}
	= \frac{n(n-1)}{2\lfloor{n/2}\rfloor\lceil{n/2}\rceil}
	\ge \frac{3}{2}.
	\]
\end{theorem}

Theorem~\ref{main1'} is a significant extension of Nam et al.'s results \cite[Theorem~1 and Corollary~1]{Nam}.
Notably, Theorem~\ref{main1'} resolves their restricted ranges of the parameters $n$ and $\eta$, and 
determines the classical and quantum optimal values in the high-privacy regime.
Note that optimal $\epsilon$-QLDP mechanisms achieving $S_\Q^\eta(n,\epsilon)$ and $A_\Q^\eta(n,\epsilon)$ in this regime are stated in the next subsection.

\subsection{Optimal $\epsilon$-QLDP mechanisms}
This subsection constructs isoclinic mechanisms, 
one of which is an optimal $\epsilon$-QLDP mechanism achieving $\OPT_n(\epsilon; \Phi_\Q)$ in the high-privacy regime.
They are constructed from \textit{equi-isoclinic tight fusion frames (EITFFs)}.
For readers' reference, we introduce EITFFs together with \textit{equi-chordal tight fusion frames (ECTFFs)}, 
whose definition is similar to that of EITFFs.
While ECTFFs and EITFFs are conventionally defined in terms of the Welch bound and the simplex bound, 
we adopt equivalent conditions as their definitions for convenience.

Recall that $n,d\ge2$ and $r\in[d-1]$ represent cardinality, dimension, and rank, respectively.

\begin{definition}[ECTFF {\cite[Theorem~5.1]{Fickus1}}]\label{ECTFF}
	We say that a set $\{P_i\}_{i\in[n]}$ of $n$ rank-$r$ orthogonal projections on $\mathbb{C}^d$ is an ECTFF 
	if it satisfies the following conditions.
	\renewcommand{\labelenumi}{$\arabic{enumi}$.}
	\begin{enumerate}
		\item
		$\sum_{i\in[n]} P_i=(nr/d)I_d$.
		\item
		There exists $c\in[0,1]$ such that $\Tr P_iP_j=rc$ for all distinct integers $i,j\in[n]$.
	\end{enumerate}
	Denote by ECTFF${}_{\mathbb{C}}(d,r,n)$ an ECTFF with dimension $d$, rank $r$, and cardinality $n$ over $\mathbb{C}$.
	By conditions~$1$--$2$, we have $c=\frac{nr-d}{d(n-1)}$.
\end{definition}

\begin{definition}[EITFF {\cite[Theorem~5.2]{Fickus1}}]\label{EITFF}
	We say that a set $\{P_i\}_{i\in[n]}$ of $n$ rank-$r$ orthogonal projections on $\mathbb{C}^d$ is an EITFF 
	if it satisfies the following conditions.
	\renewcommand{\labelenumi}{$\arabic{enumi}$.}
	\begin{enumerate}
		\item
		$\sum_{i\in[n]} P_i=(nr/d)I_d$.
		\item
		There exists $c\in[0,1]$ such that $P_jP_iP_j=cP_j$ for all distinct integers $i,j\in[n]$.
	\end{enumerate}
	Denote by EITFF${}_{\mathbb{C}}(d,r,n)$ an EITFF with dimension $d$, rank $r$, and cardinality $n$ over $\mathbb{C}$.
	By conditions~$1$--$2$, we have $c=\frac{nr-d}{d(n-1)}$.
\end{definition}

Every EITFF${}_{\mathbb{C}}(d,r,n)$ is an ECTFF${}_{\mathbb{C}}(d,r,n)$, 
but the converse does not necessarily hold.
For instance, see \cite[Example~17]{King}, which is the case $(d,r,n)=(6,3,4)$.
The second condition in Definition~\ref{EITFF} means that 
the two subspaces $\Im P_i$ and $\Im P_j$ are isoclinic.
Isoclinic subspaces are closely related to quantum error correction \cite{Kribs1,Kribs2}.

Since we use EITFFs to construct optimal $\epsilon$-QLDP mechanisms, 
we need to show their existence.
The existence of EITFFs remains mostly open \cite{Fickus2}, but the case $d=2r$ was solved as follows.

\begin{lemma}[Existence of EITFFs {\cite[p.~3203]{Fickus2}}]\label{EITFF-existence}
	The following conditions are equivalent to each other.
	\renewcommand{\labelenumi}{$\arabic{enumi}$.}
	\begin{enumerate}
		\item
		There exists an EITFF${}_{\mathbb{C}}(2r,r,n)$.
		\item
		$n\le\rho_{\mathbb{C}}(r)+2$. Here $\rho_{\mathbb{C}}(r)$ is the Radon--Hurwitz number: 
		$\rho_{\mathbb{C}}(r)=2a+2$ if $r$ is divisible by $2^a$ but not divisible by $2^{a+1}$.
	\end{enumerate}
\end{lemma}

We summarize several known results on the existence of EITFFs in Table~\ref{T1}.
While this paper generally assumes $d\ge2$ and $r\in[d-1]$, 
the trivial cases $(d,r,n)=(1,1,n),(2,2,n),\ldots$ are usually also referred to as EITFFs in the literature.
Hence, we include these trivial cases as EITFFs exclusively in Table~\ref{T1}.
By Lemma~\ref{EITFF-existence}, the condition $n\le\rho_{\mathbb{C}}(r)+2$ in Table~\ref{T1} is necessary and sufficient.
EITFFs with $r=1$ are referred to as \textit{equiangular tight frames (ETFs)}.
For details, see \cite{Fickus5}, in which many known ETFs are listed.

\begin{table}[t]
	\centering
	\caption{Existence of EITFFs.}\label{T1}
	\begin{tabular}{ccl}
		\hline
		\multirow{2}{*}{References}&Parameters for which&\multirow{2}{*}{Sufficient conditions}\\
		   &an EITFF${}_{\mathbb{C}}(d,r,n)$ exists&   \\
		\hline
		{\cite[p.~3203]{Fickus2}}&$(2r,r,n)$&$n\le\rho_{\mathbb{C}}(r)+2$\\
		{\cite[p.~194]{Fickus1}}&$\forall r\ge1$, $(d_0r,r,n)$&An EITFF${}_{\mathbb{C}}(d_0,1,n)$ exists.\\
		{\cite[Theorem~4.1]{Fickus3}}&$(n-1,2,n)$&$n$ is an odd prime power.\\
		{\cite[Theorem~4.2]{Fickus3}}&$(n,2,n)$&$n$ is a prime power $\ge4$.\\
		\multirow{2}{*}{{\cite[Example~16]{King}}}&\multirow{2}{*}{$(n(n-1)/2,n,n)$}&$n$ is a prime power, where\\
		   &   &the prime is of the form $4k-1$.\\
		\multirow{2}{*}{\cite{Fickus4}}&$(10,4,10)$, $(16,5,12)$,&   \\
		   &$(10,3,15)$&   \\
		Naimark complement&\multirow{2}{*}{$(nr-d,r,n)$}&An EITFF${}_{\mathbb{C}}(d,r,n)$ exists,\\
		{\cite[p.~2]{Fickus3}}&   &where $d\not=nr$.\\
		\hline
	\end{tabular}
\end{table}

\vskip1ex
\textbf{Isoclinic mechanisms.}---%
For an EITFF${}_{\mathbb{C}}(d,r,n)$ $\{P_i\}_{i\in[n]}$, 
we say that $\vec{\sigma}=(\sigma_1,\ldots,\sigma_n)$ is an \textit{isoclinic mechanism} if 
\begin{equation}
	\sigma_x = \frac{\mu}{d}I_d + \frac{1-\mu}{r}P_x
	\label{eq13}
\end{equation}
for every $x\in[n]$, where 
\begin{equation}
	(1-\mu)^{-1} = 1-\frac{d}{2r} + \frac{d}{2r}\sqrt{1+\frac{1-c}{\sinh^2(\epsilon/2)}}
	\quad\text{and}\quad
	c = \frac{nr-d}{d(n-1)}.
	\label{eq14}
\end{equation}
Every isoclinic mechanism is $\epsilon$-QLDP (Proposition~\ref{prop1}).
The isoclinic mechanism \eqref{eq13} can be interpreted as 
preparing the quantum states $r^{-1}P_1,\ldots,r^{-1}P_n$ and 
applying the depolarizing channel $X\mapsto\mu(\Tr X)d^{-1}I_d+(1-\mu)X$.

Consider the special case where $n\ge2$, $r=2^a$, $a\in\mathbb{Z}$, $a\ge\max\{ 0, n/2-2 \}$, and $d=2r$.
Then Lemma~\ref{EITFF-existence} implies that there exists an EITFF${}_{\mathbb{C}}(d,r,n)$.
For this EITFF, we denote by $\vec{\sigma}^*=(\sigma_1^*,\ldots,\sigma_n^*)$ the above isoclinic mechanism $\vec{\sigma}$: 
\begin{equation}
	\sigma_x^* = \frac{\mu_*}{d}I_d + \frac{1-\mu_*}{r}P_x
	\label{eq-opt-QLDP}
\end{equation}
for every $x\in[n]$, where 
\begin{equation}
	1-\mu_* = \biggl( 1+\frac{1-c}{\sinh^2(\epsilon/2)} \biggr)^{-1/2}
	\quad\text{and}\quad
	c = \frac{n-2}{2n-2}.
	\label{eq15}
\end{equation}

\begin{remark}[Reduction to classical mechanisms]\label{rem1}
	An isoclinic mechanism $\vec{\sigma}$ with $n=d/r$ reduces to a classical one.
	Indeed, $\sigma_1,\ldots,\sigma_n$ are commutative and 
	can be simultaneously diagonalized by a unitary matrix, 
	since $P_iP_jP_i=cP_i=0$ for all $i\not=j$ by \eqref{eq14}.
	Then $\sigma_x$ can be expressed as 
	\[
	\sigma_x = \frac{I_d+(e^\epsilon-1)P_x}{re^\epsilon+d-r}
	\]
	for every $x\in[n]$, which is the form that often appears in defining $\epsilon$-LDP mechanisms.
	For example, see the $\epsilon$-LDP mechanisms proposed in \cite{Nam,Park,Ye,Wang}.
\end{remark}

\vskip1ex
\textbf{Main results (3).}---%
The isoclinic mechanism $\vec{\sigma}^*$ is $\epsilon$-QLDP (see Proposition~\ref{prop1}) and 
achieves the quantum optimal value $\OPT_n(\epsilon; \Phi_\Q)$ in the high-privacy regime.

\begin{theorem}\label{main2}
	Assume all the assumptions of Theorem~$\ref{main1}$.
	Then 
	\[
	\Phi_\Q(\vec{\sigma}^*) = \phi(\mathbbm{1}_n) + \frac{\beta_0n}{4}\epsilon^2 + o(\epsilon^2)
	\]
	as $\epsilon\to+0$.
\end{theorem}

Since the Holevo information $\chi(p_{\mix}; \vec{\rho})$ in \eqref{eq09} satisfies all the assumptions of Theorem~\ref{main1}, 
the following corollary follows from Theorem~\ref{main2} immediately.

\begin{corollary}\label{coro2}
	We have 
	\[
	\chi(p_{\mix}; \vec{\sigma}^*)
	= \frac{n-1}{4n}\epsilon^2 + o(\epsilon^2)
	\]
	as $\epsilon\to+0$.
\end{corollary}

Furthermore, the isoclinic mechanism $\vec{\sigma}^*$ also achieves the quantum optimal values $S_\Q^\eta(n,\epsilon)$ and $A_\Q^\eta(n,\epsilon)$ in the high-privacy regime.

\begin{theorem}\label{main2'}
	Let $\eta\in(0,1]$.
	Then 
	\[
	S^\eta(\vec{\sigma}^*) = \frac{1}{8}\eta^2\epsilon^2 + o(\eta^2\epsilon^2)
	\quad\text{and}\quad
	A^\eta(\vec{\sigma}^*) = \frac{n-1}{4n}\eta^2\epsilon^2 + o(\eta^2\epsilon^2)
	\]
	as $\epsilon\to+0$.
\end{theorem}

\vskip1ex
\textbf{Optimality beyond the high-privacy regime.}---%
The above results hold only in the high-privacy regime.
In contrast, the isoclinic mechanism $\vec{\sigma}^*$ always achieves the quantum optimal value for the following utility function: 
\[
R_\Q^\theta(\vec{\rho}) \coloneqq
\min_{i\not=j} J_{\rho_{i,j;\theta}}^{\RLD}\Bigl[ \frac{\mathrm{d}\rho_{i,j;\theta}}{\mathrm{d}\theta}, \frac{\mathrm{d}\rho_{i,j;\theta}}{\mathrm{d}\theta} \Bigr]
= \min_{i\not=j} \Tr(\rho_j-\rho_i)^2\rho_{i,j;\theta}^{-1},
\]
where $\rho_{i,j;\theta}=(1-\theta)\rho_i+\theta\rho_j$ and $\theta\in[0,1]$.
This is the minimum among the pairwise RLD Fisher information values of the one-parameter family $(\rho_{i,j;\theta})_{\theta\in[0,1]}$ at the point $\theta$, and 
is the same as the utility function considered in \cite{Yoshida}.
The exact value of $\OPT_n(\epsilon; R_\Q^\theta)$ was determined in \cite[Theorem~2.2]{Yoshida}, but 
no optimal $\epsilon$-QLDP mechanism achieving this value was previously known.
We can show that the isoclinic mechanism $\vec{\sigma}^*$ achieves this value, 
which is proved in Appendix~\ref{appG}

\begin{theorem}\label{main5}
	Let $\theta\in[0,1]$. Then 
	\[
	R_\Q^\theta(\vec{\sigma}^*)
	= \frac{(e^\epsilon-1)^2}{( (1-\theta)e^\epsilon+\theta )( \theta e^\epsilon+1-\theta )}
	= \OPT_n(\epsilon; R_\Q^\theta).
	\]
\end{theorem}

The quantum optimal value $\OPT_n(\epsilon; R_\Q^\theta)$ is independent of $n$.
In contrast, its classical counterpart strictly decreases in $n$, as shown in \cite[Eq.~(22) and Theorem~2.2]{Yoshida}.
These observations yield a quantum advantage for the utility function $R_\Q^\theta$ (and its classical counterpart).
Hence, this quantum advantage persists even in the low-privacy regime (i.e., for large $\epsilon$).
On the other hand, the argument in Section~\ref{origin} suggests that 
isoclinic mechanisms show no quantum advantage in this regime for the error exponents in symmetric and asymmetric hypothesis testing.

\subsection{Optimal $\epsilon$-LDP mechanisms}
This subsection introduces binary mechanisms, which were proposed in \cite[Section~3.2]{Kairouz}.
We say that a column-stochastic matrix $b\in\mathbb{R}^{2\times n}$ is a binary mechanism if 
there exists a disjoint union $[n]=A_1\cup A_2$ such that 
\[
b(y|x) = \frac{1+(e^\epsilon-1)\mathbbm{1}_{A_y}(x)}{e^\epsilon+1}
\]
for every $x\in[n]$ and every $y\in[2]$, where $\mathbbm{1}_A$ denotes the indicator function of a set $A$.
When $A_1=\{1,2,\ldots,\lfloor{n/2}\rfloor\}$ and $A_2=\{\lfloor{n/2}\rfloor+1,\ldots,n-1,n\}$, 
we denote by $b^*$ the binary mechanism $b$.

\vskip1ex
\textbf{Main results (4).}---%
The binary mechanism $b^*$ achieves the classical optimal values $\OPT_n(\epsilon; \Phi_\C)$ in the high-privacy regime.
As a corollary, the mutual information (or the Holevo information) $\chi(p_{\mix}; q)$ is optimized by $b^*$ in this regime.
These results were already suggested in \cite[Theorems~5 and 6]{Kairouz}, but 
it was not proved for a general utility function of the form \eqref{eq-sum-subl}.
While block design mechanisms \cite{Park,Nam} may be considered as candidates for optimal $\epsilon$-LDP mechanisms, 
we do not employ them in this paper due to the simplicity of the binary mechanism $b^*$.
Furthermore, the binary mechanism $b^*$ achieves the classical optimal values $S_\C^\eta(n,\epsilon)$ and $A_\C^\eta(n,\epsilon)$ in this regime.
We summarize the above results as follows.

\begin{theorem}\label{main3}
	Assume assumptions~$1$--$2$ of Theorem~$\ref{main1}$.
	Then 
	\[
	\Phi_\C(b^*) = \phi(\mathbbm{1}_n) + \frac{\lfloor{n/2}\rfloor\lceil{n/2}\rceil}{n-1}\cdot\frac{\beta_0}{2}\epsilon^2 + o(\epsilon^2)
	\]
	as $\epsilon\to+0$.
\end{theorem}

\begin{corollary}\label{coro3}
	We have 
	\[
	\chi(p_{\mix}; b^*)
	= \frac{\lfloor{n/2}\rfloor\lceil{n/2}\rceil}{2n^2}\epsilon^2 + o(\epsilon^2)
	\]
	as $\epsilon\to+0$.
\end{corollary}

\begin{theorem}\label{main3'}
	Let $\eta\in(0,1]$.
	Then 
	\[
	S^\eta(b^*) = \frac{\lfloor{n/2}\rfloor\lceil{n/2}\rceil}{4n(n-1)}\eta^2\epsilon^2 + o(\eta^2\epsilon^2)
	\quad\text{and}\quad
	A^\eta(b^*) = \frac{\lfloor{n/2}\rfloor\lceil{n/2}\rceil}{2n^2}\eta^2\epsilon^2 + o(\eta^2\epsilon^2)
	\]
	as $\epsilon\to+0$.
\end{theorem}

\section{Exact values of $S_\C^1(n,\epsilon)$, $A_\C^\eta(n,\epsilon)$, $S^\eta(\vec{\sigma})$, and $A^\eta(\vec{\sigma})$}\label{exact}
As before, let $d,n\ge2$, $r\in[d-1]$, and $\epsilon>0$.
In this section, we clarify the exact values of $S_\C^1(n,\epsilon)$, $A_\C^\eta(n,\epsilon)$, $S^\eta(\vec{\sigma})$, and $A^\eta(\vec{\sigma})$, 
where $\vec{\sigma}$ is an isoclinic mechanism in \eqref{eq13}.
These values are used in Section~\ref{origin} to discuss quantum advantages not only in the high-privacy regime but also in a moderate-privacy regime.
In contrast, the exact value of $S_\C^\eta(n,\epsilon)$ is unknown for every $\eta\in(0,1)$ and known only for $\eta=1$.
Hence, we state upper and lower bounds for $S_\C^\eta(n,\epsilon)$.
We begin with the classical optimal values $S_\C^\eta(n,\epsilon)$ and $A_\C^\eta(n,\epsilon)$.

\begin{lemma}\label{lemC1}
	Let $\eta\in(0,1]$. Then 
	\begin{gather*}
		\underline{S}_\C^\eta(n,\epsilon) \le S_\C^\eta(n,\epsilon) \le \overline{S}_\C^\eta(n,\epsilon)
		\quad\text{and}\\
		A_\C^\eta(n,\epsilon)
		= \max_{k\in[n-1]} \Bigl( \frac{k}{n}L(\Delta_{k/n}^\eta(1)) + \frac{n-k}{n}L(\Delta_{k/n}^\eta(0)) \Bigr),
	\end{gather*}
	where $L(t)=t\ln t$, 
	\begin{align*}
		\overline{S}_\C^\eta(n,\epsilon)
		&= -\ln\biggl( 1 - \frac{n+\eta^2-1}{n}\max_{k\in[n-1]} \frac{k(n-k)}{n(n-1)}\Bigl( \sqrt{\Delta_{k/n}^1(1)} - \sqrt{\Delta_{k/n}^1(0)} \Bigr)^2 \biggr),\\
		\underline{S}_\C^\eta(n,\epsilon)
		&= -\ln\biggl( 1 - \max_{k\in[n-1]} \frac{k(n-k)}{n(n-1)}\Bigl( \sqrt{\Delta_{k/n}^\eta(1)} - \sqrt{\Delta_{k/n}^\eta(0)} \Bigr)^2 \biggr),\quad\text{and}\\
		\Delta_u^\eta(t)
		&= \frac{\eta(1+(e^\epsilon-1)t)}{u(e^\epsilon-1)+1} + 1-\eta.
	\end{align*}
\end{lemma}

Lemma~\ref{lemC1} except for the inequality $\underline{S}_\C^\eta(n,\epsilon) \le S_\C^\eta(n,\epsilon)$ was proved in \cite[Proposition~4]{Nam}.
At first glance, the expressions of $\overline{S}_\C^\eta(n,\epsilon)$ and $A_\C^\eta(n,\epsilon)$ are different from the original ones, but 
they are equal to each other, respectively.
Hence, regarding Lemma~\ref{lemC1}, we only prove the inequality $\underline{S}_\C^\eta(n,\epsilon) \le S_\C^\eta(n,\epsilon)$ in Appendix~\ref{appD}.
We expect that $S_\C^\eta(n,\epsilon)=\underline{S}_\C^\eta(n,\epsilon)$ for every $\eta\in(0,1)$, but 
it remains open.

The expression of $A_\C^\eta(n,\epsilon)$ can be generalized as follows: 
\begin{equation}
	\sup_{q:\,\text{$\epsilon$-LDP}} \min_{x\in[n]} D_F(\Diag_x(\tilde{q})\|\Diag_{\avg}(q))
	= \max_{k\in[n-1]} \Bigl( \frac{k}{n}F(\Delta_{k/n}^\eta(1)) + \frac{n-k}{n}F(\Delta_{k/n}^\eta(0)) \Bigr),
	\label{eq52}
\end{equation}
where $F\colon (0,\infty)\to\mathbb{R}$ is a convex function satisfying $F(1)=0$.
By substituting $F=L$ into \eqref{eq52}, the expression of $A_\C^\eta(n,\epsilon)$ can be recovered.
In light of this generalization, regarding $A_\C^\eta(n,\epsilon)$, 
it is natural to restate \cite[Proposition~4]{Nam} as Lemma~\ref{lemC1}.

By $\underline{S}_\C^1(n,\epsilon)=\overline{S}_\C^1(n,\epsilon)$, 
Lemma~\ref{lemC1} reduces to the following lemma if $\eta=1$ \cite[Proposition~4]{Nam}.

\begin{lemma}\label{lemC2}
	Let $L(t)=t\ln t$. Then 
	\begin{align*}
		S_\C^1(n,\epsilon) &= -\ln\Bigl( 1 - \frac{(e^{\epsilon/2}-1)^2}{n-1}\max_{k\in[n-1]} \frac{k(n-k)}{ke^\epsilon+n-k} \Bigr)\quad\text{and}\\
		A_\C^1(n,\epsilon) &= \max_{k\in[n-1]} \frac{kL(e^\epsilon)-nL\bigl( (ke^\epsilon+n-k)n^{-1} \bigr)}{ke^\epsilon+n-k}.
	\end{align*}
\end{lemma}

Next, we state the exact values of $S^\eta(\vec{\sigma})$ and $A^\eta(\vec{\sigma})$ for every isoclinic mechanism $\vec{\sigma}$ in \eqref{eq13}.
Recall that $\tilde{\rho}_k$ is defined as \eqref{eq11} for an $\epsilon$-QLDP mechanism $\vec{\rho}=(\rho_1,\ldots,\rho_n)$ and an integer $k\in[n]$.
Hence, for every $k\in[n]$, 
\begin{align}
	\tilde{\sigma}_k &= \sum_{x\in[n]} p_k^\eta(x)\sigma_x
	= \eta\sigma_k + \frac{1-\eta}{n}\sum_{x\in[n]} \sigma_x
	= \eta\sigma_k + \frac{1-\eta}{d}I_d\nonumber\\
	&= \frac{\mu^{[\eta]}}{d}I_d + \frac{1-\mu^{[\eta]}}{r}P_k,\label{eq16}
\end{align}
where $\mu^{[\eta]}$ is defined as $\mu^{[\eta]}=\eta\mu+1-\eta$.
This is the same as the quantum state of replacing $\mu$ with $\mu^{[\eta]}$, in the definition of $\sigma_k$.
The following lemma is proved in Appendix~\ref{appD} by direct calculation.

\begin{lemma}\label{lemQ1}
	Let $\vec{\sigma}$ be an $\epsilon$-QLDP mechanism in \eqref{eq13}, 
	where the condition on $\mu$ in \eqref{eq14} can be omitted here.
	Let $\eta\in(0,1]$, $\mu^{[\eta]}=\eta\mu+1-\eta$, $L(t)=t\ln t$,  
	\begin{align*}
		G_\sym(t) &= (1-c)(\sqrt{ut+1-t}-\sqrt{ut})^2,\quad\text{and}\\
		G_\asym(t) &= uL(t+u^{-1}(1-t)) + (1-u)L(t),
	\end{align*}
	where $u=r/d$ and $c=\frac{nu-1}{n-1}$.
	Then 
	\[
	S^\eta(\vec{\sigma}) = -\ln\bigl( 1-G_\sym(\mu^{[\eta]}) \bigr)
	\quad\text{and}\quad
	A^\eta(\vec{\sigma}) = G_\asym(\mu^{[\eta]}).
	\]
	In particular, the equalities 
	\[
	S^\eta(\vec{\sigma}^*) = -\ln\bigl( 1-G_\sym^*(\mu_*^{[\eta]}) \bigr)
	\quad\text{and}\quad 
	A^\eta(\vec{\sigma}^*) = G_\asym^*(\mu_*^{[\eta]})
	\]
	hold, where $\mu_*^{[\eta]}=\eta\mu_*+1-\eta$, $c=\frac{n-2}{2n-2}$,
	\[
	G_\sym^*(t) = (1-c)\Bigl( 1-\sqrt{t(2-t)} \Bigr),
	\quad\text{and}\quad
	G_\asym^*(t) = \frac{L(2-t)+L(t)}{2}.
	\]
\end{lemma}

\vskip1ex
\textbf{Numerical computation.}---%
Figure~\ref{Fig1} is the graphs of the two ratios $S^1(\vec{\sigma}^*)/S_\C^1(n,\epsilon)$ and $A^1(\vec{\sigma}^*)/A_\C^1(n,\epsilon)$.
There is a quantum advantage if the ratio in the symmetric or asymmetric case is greater than one.
One can observe quantum advantages for $\epsilon\in(0,1]$ and $n\ge3$ in both cases.
Moreover, the symmetric case exhibits a quantum advantage in a wider range of the parameter $\epsilon$.
In particular, one can find a quantum advantage in the symmetric case for $\epsilon\in(0,2]$ and $n\ge4$.

\begin{figure}[t]
	\centering
	\includegraphics[scale=0.5]{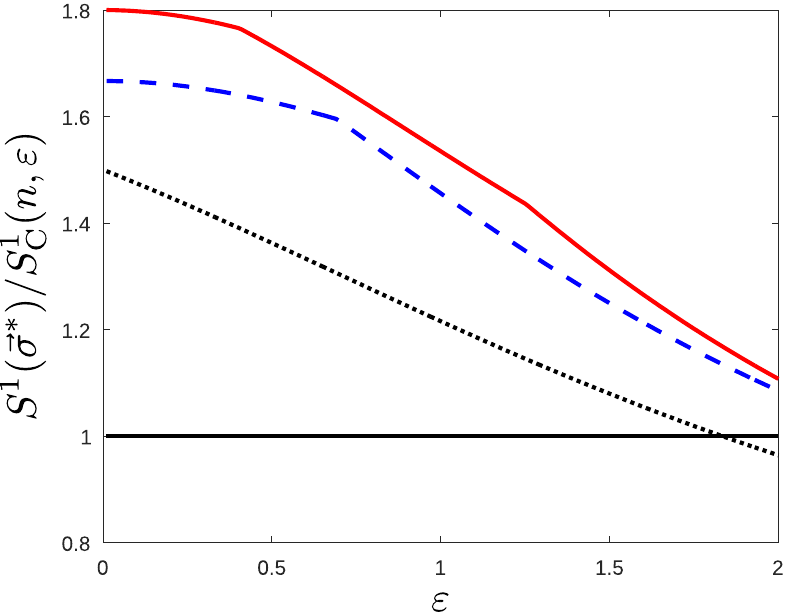}
	\hspace{1em}
	\includegraphics[scale=0.5]{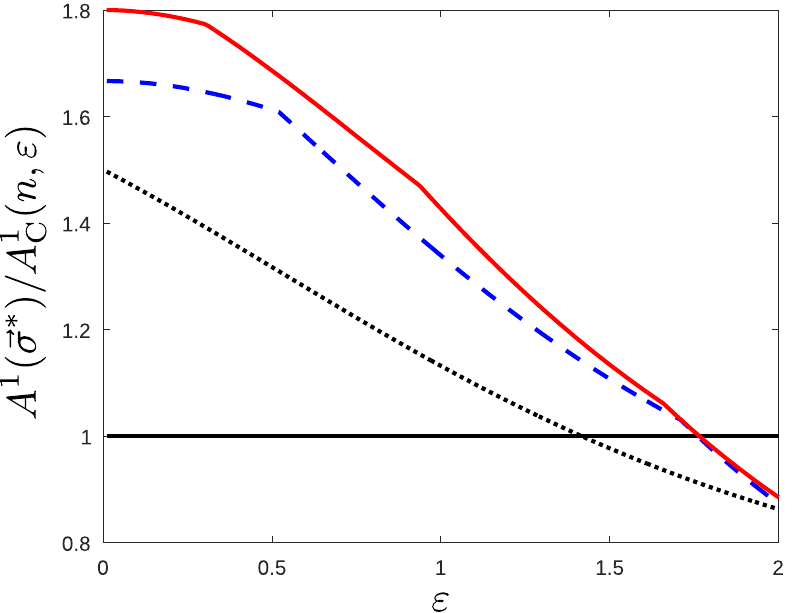}
	\caption{Graphs of $S^1(\vec{\sigma}^*)/S_\C^1(n,\epsilon)$ and $A^1(\vec{\sigma}^*)/A_\C^1(n,\epsilon)$.
	The horizontal axes represent the values of $\epsilon$, and 
	the vertical axes represent the values of the two ratios.
	The black dotted lines, the blue dashed lines, and the red solid lines are the cases $n=3$, $n=6$, and $n=10$, respectively.}
	\label{Fig1}
\end{figure}

\section{Structural origin of quantum advantage}\label{origin}
As before, let $d,n\ge2$, $r\in[d-1]$, and $\epsilon>0$.
This section discusses why isoclinic mechanisms outperform classical ones.
We consider the following two reasons to be crucial: 
\begin{itemize}
	\item
	the depolarizing noise parameter $\mu$ is smaller than its classical counterpart, and 
	\item
	the set of values $r/d$ is large enough to include $k/n$ for every integer $k\in[1,n/2]$,
\end{itemize}
where $\mu$, $r$, and $d$ are the parameters used in defining an isoclinic mechanism in \eqref{eq13}.
The first one leads to quantum advantages via the data processing inequality.
Although the second one is currently no more than a conjecture, 
it naturally emerges by investigating the exact values given in Section~\ref{exact}.
Assuming this conjecture, we investigate the ranges of the parameter $\epsilon$ 
for which quantum advantages arise in symmetric and asymmetric hypothesis testing, and 
show quantum advantages not only in the high-privacy regime but also in a moderate-privacy regime.

\vskip1ex
\textbf{Depolarizing noise parameter.}---%
As stated immediately after \eqref{eq14}, 
an isoclinic mechanism can be interpreted as 
preparing certain $n$ quantum states and applying a depolarizing channel.
For convenience, we use the parameter $u=r/d$ as in Lemma~\ref{lemQ1}, 
which satisfies $1/n\le u\le1/2$ because 
(i) $c\ge0$ and (ii) $2r=\dim(\Im P_1+\Im P_2)\le d$.
Regard 
\begin{equation}
	c = \frac{nu-1}{n-1}
	\quad\text{and}\quad 
	(1-\mu)^{-1} = 1-\frac{1}{2u} + \frac{1}{2u}\sqrt{1+\frac{1-c}{\sinh^2(\epsilon/2)}}
	\label{eq17}
\end{equation}
as functions of $u\in[1/n,1/2]$.
Then the depolarizing noise parameter in \eqref{eq14} can be expressed as $\mu(r/d)$.

The classical counterpart of $\mu(r/d)$ can be written as $\mu_0(k/n)$, 
where $\mu_0=\mu_0(u)$ is the following function of $u\in(0,1)$: 
\begin{equation}
	\mu_0 = \frac{1}{u(e^\epsilon-1)+1}.
	\label{eq38}
\end{equation}
For details, see Appendix~\ref{appE} (and see also Remark~\ref{rem1}, Lemmas~\ref{lemC1} and \ref{lemC2}).
From $c\ge0$, it follows that 
\begin{align}
	(1-\mu)^{-1} &\le 1-\frac{1}{2u} + \frac{1}{2u}\sqrt{1+\frac{1}{\sinh^2(\epsilon/2)}}
	= 1-\frac{1}{2u} + \frac{1}{2u}\cdot\frac{\cosh(\epsilon/2)}{\sinh(\epsilon/2)} \label{eq59}\\
	&= 1+\frac{e^{-\epsilon/2}}{2u\sinh(\epsilon/2)}
	= 1+\frac{1}{u(e^\epsilon-1)}
	= (1-\mu_0)^{-1}. \label{eq43}
\end{align}
By this and $\mu,\mu_0\in(0,1)$, 
the inequality $\mu(u)\le\mu_0(u)$ holds for every $u\in[1/n,1/2]$.
Hence, the depolarizing noise parameter is smaller than its classical counterpart in general.

\vskip1ex
\textbf{Set of values $r/d$.}---%
We next introduce the set of values $r/d$.
Let $\cQ_n$ be the set of values $r/d$ such that an EITFF${}_{\mathbb{C}}(d,r,n)$ exists.
By \cite[Corollary~3.4]{Iverson} and Table~\ref{T1}, 
one can represent $\cQ_n$ explicitly for small $n$: 
\begin{gather*}
	\cQ_2 = \{1/2\},\quad
	\cQ_3 = \{1/3,1/2\},\quad
	\cQ_4 = \{1/4,1/3,1/2\},\quad\text{and}\\ 
	\cQ_5 = \{1/5,1/4,1/3,2/5,1/2\}.
\end{gather*}
These sets can be expressed as $\bigcup_{m=2}^n \cC_m$, where $\cC_m = \{ k/m : k\in[1,m/2]\cap\mathbb{Z} \}$.
However, this is not the case for $n=12$, since an EITFF${}_{\mathbb{C}}(16,5,12)$ exists (see Table~\ref{T1}).

\vskip1ex
\textbf{Expressions of several optimal values.}---%
As before, regard $c$, $\mu$, and $\mu_0$ as functions of $u$.
For $u\in[1/n,1/2]$ and $\eta\in(0,1]$, 
define $G_\sym^{\iso;\eta}(u)$ and $G_\asym^{\iso;\eta}(u)$ as\footnote{\label{fnote1}
Since $G_\sym(t)$ and $G_\asym(t)$ depend on $u$, 
they can be written as $G_\sym(t,u)$ and $G_\asym(t,u)$, respectively, to avoid any misunderstanding.
Then \eqref{eq-G-iso} can be expressed as $G_\sym^{\iso;\eta}(u) = G_\sym(\mu^{[\eta]}(u),u)$ and $G_\asym^{\iso;\eta}(u) = G_\asym(\mu^{[\eta]}(u),u)$.
Similarly, \eqref{eq-G-C} can be expressed as $G_\sym^{\C;\eta}(u) = G_\sym(\mu_0^{[\eta]}(u),u)$ and $G_\asym^{\C;\eta}(u) = G_\asym(\mu_0^{[\eta]}(u),u)$.} 
\begin{equation}
	G_\sym^{\iso;\eta}(u) = G_\sym(\mu^{[\eta]})
	\quad\text{and}\quad
	G_\asym^{\iso;\eta}(u) = G_\asym(\mu^{[\eta]}),
	\label{eq-G-iso}
\end{equation}
where $\mu^{[\eta]}$, $G_\sym$, and $G_\asym$ are defined in Lemma~\ref{lemQ1}.
Then Lemma~\ref{lemQ1} yields 
\begin{equation}
\begin{split}
	S_\iso^\eta(n,\epsilon)
	&\coloneqq \sup_{\vec{\sigma}} S^\eta(\vec{\sigma})
	= \sup_{u\in\cQ_n} (-1)\ln( 1-G_\sym^{\iso;\eta}(u) )
	\quad\text{and}\\
	A_\iso^\eta(n,\epsilon)
	&\coloneqq \sup_{\vec{\sigma}} A^\eta(\vec{\sigma})
	= \sup_{u\in\cQ_n} G_\asym^{\iso;\eta}(u)
\end{split}\label{eq55}
\end{equation}
for every $\eta\in(0,1]$, 
where the two suprema $\sup_{\vec{\sigma}}$ are taken over all isoclinic mechanisms $\vec{\sigma}$.
Also, we prove the following equalities in Appendix~\ref{appE}: 
\begin{equation}
	\underline{S}_\C^\eta(n,\epsilon) = \max_{u\in\cC_n} (-1)\ln( 1-G_\sym^{\C;\eta}(u) )
	\quad\text{and}\quad
	A_\C^\eta(n,\epsilon) = \max_{u\in\cC_n} G_\asym^{\C;\eta}(u),
	\label{eq56}
\end{equation}
for every $\eta\in(0,1]$, 
where $G_\sym^{\C;\eta}(u)$ and $G_\asym^{\C;\eta}(u)$ are defined as\footref{fnote1} 
\begin{equation}
	G_\sym^{\C;\eta}(u) = G_\sym(\mu_0^{[\eta]})
	\quad\text{and}\quad
	G_\asym^{\C;\eta}(u) = G_\asym(\mu_0^{[\eta]})
	\label{eq-G-C}
\end{equation}
for $u\in(0,1)$ and $\eta\in(0,1]$.
Equations~\eqref{eq56} and \eqref{eq-G-C} justify calling $\mu_0$ the classical counterpart of $\mu$.
Furthermore, we prove the following inequalities in Appendix~\ref{appE}: 
\begin{equation}
	S_\iso^\eta(n,\epsilon) \ge \sup_{u\in\cQ_n} (-1)\ln( 1-G_\sym^{\C;\eta}(u) )
	\quad\text{and}\quad
	A_\iso^\eta(n,\epsilon) \ge \sup_{u\in\cQ_n} G_\asym^{\C;\eta}(u).
	\label{eq58}
\end{equation}

\vskip1ex
\textbf{Saturation hypothesis.}---%
We now compare \eqref{eq56} and \eqref{eq58}.
Since we expect that isoclinic mechanisms outperform classical ones 
(not only in the high-privacy regime but also in a moderate-privacy regime), 
it is natural to conjecture $\cC_n\subset\cQ_n$, 
which implies $\underline{S}_\C^\eta(n,\epsilon) \le S_\iso^\eta(n,\epsilon)$ and 
$A_\C^\eta(n,\epsilon) \le A_\iso^\eta(n,\epsilon)$.
To the best of our knowledge, it remains open whether $\cC_n\subset\cQ_n$.
In Propositions~\ref{prop2}, \ref{prop3}, and Theorem~\ref{main4} below, 
we assume $\cC_n\subset\cQ_n$ and refer to it as the \textit{saturation hypothesis}.
This hypothesis is actually true if $n\in\{2,3,4,5,6\}$.
The following proposition summarizes the above argument.

\begin{proposition}\label{prop2}
	Let $\eta\in(0,1]$. Assume the saturation hypothesis.
	Then $\underline{S}_\C^\eta(n,\epsilon) \le S_\iso^\eta(n,\epsilon)$ and 
	$A_\C^\eta(n,\epsilon) \le A_\iso^\eta(n,\epsilon)$.
	In particular, $S_\C^1(n,\epsilon) \le S_\iso^1(n,\epsilon)$.
\end{proposition}

While we prove Proposition~\ref{prop2} (more precisely, \eqref{eq56} and \eqref{eq58}) in Appendix~\ref{appE}, 
the methodology suggests a further generalization.
Also, we do not know whether the equalities $S_\iso^\eta(n,\epsilon)=S_\Q^\eta(n,\epsilon)$ and $A_\iso^\eta(n,\epsilon)=A_\Q^\eta(n,\epsilon)$ hold.
Verifying these equalities, as well as the saturation hypothesis, is left as an important open problem.

\vskip1ex
\textbf{Range of parameter $\epsilon$ for which quantum advantages arise.}---%
Assuming the saturation hypothesis, 
we can investigate the ranges of the parameter $\epsilon$ such that 
$S_\C^1(n,\epsilon)<S_\iso^1(n,\epsilon)$ and $A_\C^1(n,\epsilon)<A_\iso^1(n,\epsilon)$ as follows.
Let $n\ge3$, $\eta\in(0,1]$, and $u\in(1/n,1/2]$ here.
Then $1-c<1$ and thus, inequality \eqref{eq59} is strict, i.e., $\mu<\mu_0$.
Since the functions $G_\sym(t)$ and $G_\asym(t)$ in Lemma~\ref{lemQ1} strictly decrease in $t\in[0,1]$, 
it turns out that 
\begin{equation}
	G_\sym^{\iso;\eta}(u) > G_\sym^{\C;\eta}(u)
	\quad\text{and}\quad
	G_\asym^{\iso;\eta}(u) > G_\asym^{\C;\eta}(u).
	\label{eq60}
\end{equation}
Let $u_*$ be the value of $u$ that achieves the first maximum in \eqref{eq56}.
If $u_*\not=1/n$, then the saturation hypothesis, \eqref{eq55}, and \eqref{eq60} imply 
\begin{equation*}
	S_\iso^\eta(n,\epsilon)
	\ge -\ln( 1-G_\sym^{\iso;\eta}(u_*) )
	> -\ln( 1-G_\sym^{\C;\eta}(u_*) )
	= \underline{S}_\C^\eta(n,\epsilon).
\end{equation*}
If $u_*=1/n$ and $G_\sym^{\iso;\eta}(1/(n-1)) > G_\sym^{\C;\eta}(1/n)$, 
then \eqref{eq55} and \eqref{eq60} imply 
\begin{equation*}
	S_\iso^\eta(n,\epsilon)
	\ge -\ln\bigl( 1-G_\sym^{\iso;\eta}(1/(n-1)) \bigr)
	> -\ln( 1-G_\sym^{\C;\eta}(1/n) )
	= \underline{S}_\C^\eta(n,\epsilon),
\end{equation*}
where the first inequality follows from $1/(n-1)\in\cQ_n$, 
which is based on the fact that there exist $n$ unit vectors in $\mathbb{R}^{n-1}$ forming the vertices of a regular simplex centered at the origin.
For the asymmetric case, the above argument also works well.
Consequently, we obtain the following proposition.

\begin{proposition}\label{prop3}
	Let $\eta\in(0,1]$. Assume $n\ge3$ and the saturation hypothesis.
	If $G_\sym^{\iso;\eta}(1/(n-1)) > G_\sym^{\C;\eta}(1/n)$, then $\underline{S}_\C^\eta(n,\epsilon)<S_\iso^\eta(n,\epsilon)$.
	If $G_\asym^{\iso;\eta}(1/(n-1)) > G_\asym^{\C;\eta}(1/n)$, then $A_\C^\eta(n,\epsilon)<A_\iso^\eta(n,\epsilon)$.
\end{proposition}

If $\eta=1$, we can show the following theorem simpler than Proposition~\ref{prop3}, 
which is proved in Appendix~\ref{appF}.

\begin{theorem}\label{main4}
	Assume $n\ge3$ and the saturation hypothesis.
	If $\epsilon\le\ln(n-1)(n-2)$, then $S_\C^1(n,\epsilon)<S_\iso^1(n,\epsilon)$.
	If $\epsilon\le(1/2)\ln(n-1)(n-2)$, then $A_\C^1(n,\epsilon)<A_\iso^1(n,\epsilon)$.
\end{theorem}

\vskip1ex
\textbf{Numerical computation.}---%
The two ratios 
$S_\iso^1(n,\epsilon)/S_\C^1(n,\epsilon)$ and $A_\iso^1(n,\epsilon)/A_\C^1(n,\epsilon)$ 
are plotted in Figure~\ref{Fig2} for small $n$.
Both ratios tend to decrease as $\epsilon$ increases, and they eventually stabilize at $1$.
Notably, the numerical results in Figure~\ref{Fig2} suggest that there is no quantum advantage in the low-privacy regime (i.e., for large $\epsilon$).
To rigorously establish the absence of quantum advantages in this regime, 
we need to prove the optimality of isoclinic mechanisms in this regime.

\begin{figure}[t]
	\centering
	\includegraphics[scale=0.5]{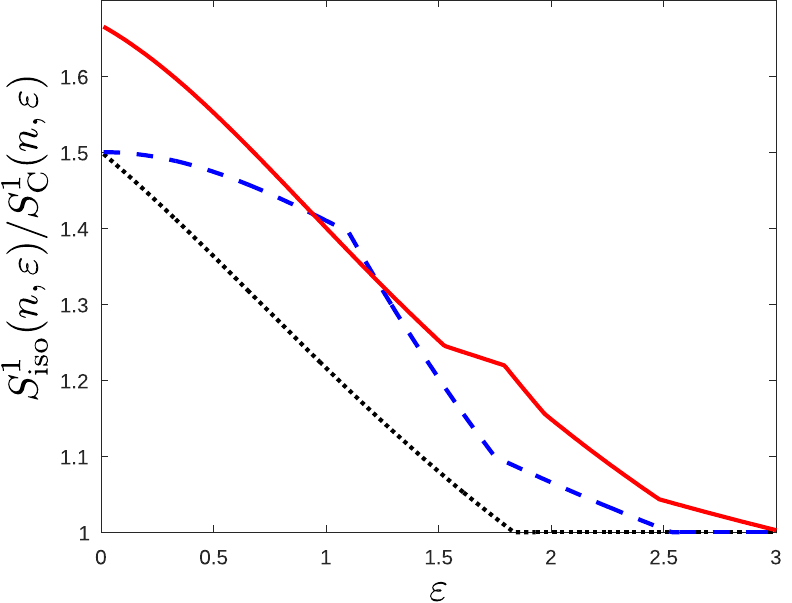}
	\hspace{1em}
	\includegraphics[scale=0.5]{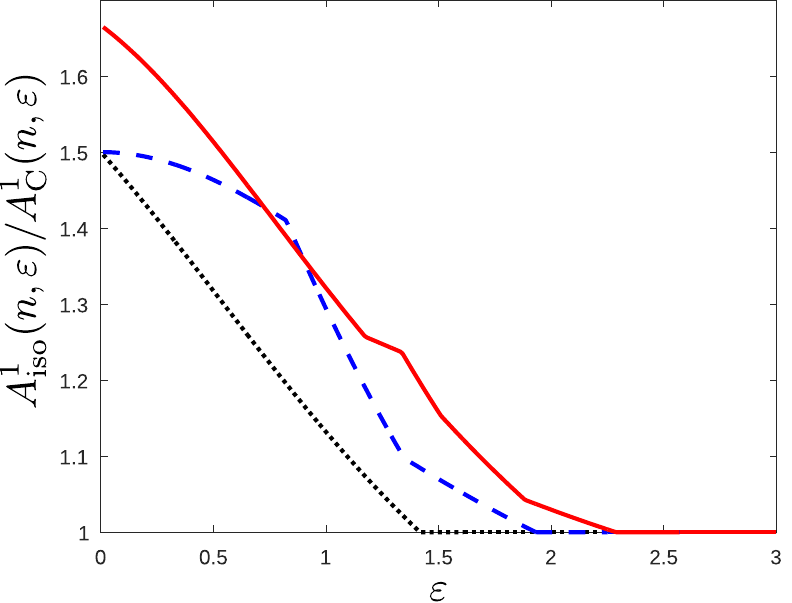}
	\caption{Graphs of $S_\iso^1(n,\epsilon)/S_\C^1(n,\epsilon)$ and $A_\iso^1(n,\epsilon)/A_\C^1(n,\epsilon)$.
	The horizontal axes represent the values of $\epsilon$, and 
	the vertical axes represent the values of the two ratios.
	The black dotted lines, the blue dashed lines, and the red solid lines are the cases $n=3$, $n=4$, and $n=5$, respectively.}
	\label{Fig2}
\end{figure}

\section{Conclusion and future work}\label{conc}
We have constructed the isoclinic mechanism $\vec{\sigma}^*$, and 
shown that it is an optimal $\epsilon$-QLDP mechanism achieving the quantum optimal value $\OPT_n(\epsilon; \Phi_\Q)$ in the high-privacy regime.
For the isoclinic mechanism $\vec{\sigma}^*$, the minimum dimension of the underlying space $\mathbb{C}^d$ is $2^{a+1}$, 
where $a=\max\{ 0,\lceil{n/2}\rceil-2 \}$.
We do not know whether this is the minimum dimension among the underlying spaces of all optimal $\epsilon$-QLDP mechanisms in this regime.
Determining this minimum dimension is an interesting direction for future work.
For comparison, the binary mechanism has the underlying space $\mathbb{R}^2$.
Since we have shown that the binary mechanism is an optimal $\epsilon$-LDP mechanism 
achieving the classical optimal value $\OPT_n(\epsilon; \Phi_\C)$ in this regime, 
the minimum dimension among the underlying spaces of all optimal $\epsilon$-LDP mechanisms in this regime is two.

We have derived the asymptotic ratio \eqref{eq02} of the classical and quantum optimal values.
This ratio is essentially the same as the lower bound in \cite[p.~5]{Yoshida}.
We explain the reason below.
According to \cite[p.~5]{Yoshida}, the inequality 
\begin{equation}
	\sqrt{\frac{n(n-1)}{2\lfloor{n/2}\rfloor\lceil{n/2}\rceil}} \le \frac{e^{\epsilon_{\inf}}-1}{e^\epsilon-1} \le n-1
	\label{eq18}
\end{equation}
holds, where $\epsilon_{\inf}=\epsilon_{\inf}(n,\epsilon)$ is the infimum of $\epsilon'>0$ satisfying the following condition: 
for every $\epsilon$-QLDP mechanism $\vec{\rho}$, there exist an $\epsilon'$-LDP mechanism $q$ and a quantum channel $\Lambda$ such that 
$\rho_x=\Lambda(\Diag_x(q))$ for every $x\in[n]$.
Assume that the utility functions $\Phi_\C$ and $\Phi_\Q$ satisfy the data processing inequality, 
all the assumptions of Theorem~\ref{main1}, and $\phi(\mathbbm{1}_n)=0$.
Then the first inequality of \eqref{eq18} follows from Theorems~\ref{main1} and \ref{main2} as $\epsilon\to+0$.
Indeed, the definition of $\epsilon_{\inf}$ and the data processing inequality imply\footnote{
Strictly speaking, Lemma~\ref{lemKairouz} is also needed to show \eqref{eq19}, 
since one can verify using this lemma that $\inf_{\epsilon'>\epsilon_{\inf}} \OPT_n(\epsilon'; \Phi_\C) = \OPT_n(\epsilon_{\inf}; \Phi_\C)$.} 
\begin{equation}
	\frac{\Phi_\Q(\vec{\sigma}^*)}{\OPT_n(\epsilon; \Phi_\C)}
	\le \frac{\OPT_n(\epsilon_{\inf}; \Phi_\C)}{\OPT_n(\epsilon; \Phi_\C)}.
	\label{eq19}
\end{equation}
Also, $\epsilon_{\inf}=O(\epsilon)$ as $\epsilon\to+0$ by the second inequality of \eqref{eq18}.
This, \eqref{eq19}, and Theorems~\ref{main1} and \ref{main2} imply 
\[
\sqrt{\frac{n(n-1)}{2\lfloor{n/2}\rfloor\lceil{n/2}\rceil}}
\le \liminf_{\epsilon\to+0} \frac{\epsilon_{\inf}}{\epsilon}
= \liminf_{\epsilon\to+0} \frac{e^{\epsilon_{\inf}}-1}{e^\epsilon-1}.
\]
Since the isoclinic mechanism $\vec{\sigma}^*$ is an optimal $\epsilon$-QLDP mechanism, 
this argument cannot yield a stronger lower bound.
Therefore, we conjecture that the lower bound is tight: 
\[
\liminf_{\epsilon\to+0} \frac{e^{\epsilon_{\inf}}-1}{e^\epsilon-1} = \sqrt{\frac{n(n-1)}{2\lfloor{n/2}\rfloor\lceil{n/2}\rceil}}.
\]
This stands in contrast to our previous expectation \cite[Section~6]{Yoshida}.
Proving or disproving the above tightness is also an interesting direction for future work.

While Kairouz et al.\ \cite{Kairouz} provided 
\renewcommand{\theenumi}{C\arabic{enumi}}
\begin{enumerate}
	\item\label{C1}
	a class of LDP mechanisms and 
	\item\label{C2}
	suitable assumptions on the utility function $\Phi_\C$ 
\end{enumerate}
such that a maximizer of $\Phi_\C$ under the LDP constraint exists in this class, 
our long-term goal is to establish a quantum counterpart of the results of \cite{Kairouz}.
To achieve this goal, we need to find 
\renewcommand{\theenumi}{Q\arabic{enumi}}
\begin{enumerate}
	\item\label{Q1}
	a class of QLDP mechanisms and 
	\item\label{Q2}
	suitable assumptions on the utility function $\Phi_\Q$ 
\end{enumerate}
such that a maximizer of $\Phi_\Q$ under the QLDP constraint exists in this class.
The class of all isoclinic mechanisms is a promising candidate for \eqref{Q1}.
However, \eqref{Q2} remains open and should be further studied.
We hope that our work will stimulate further research toward establishing a quantum counterpart of the results of \cite{Kairouz}, 
potentially leading to a breakthrough in the study of QLDP.

We have shown that isoclinic mechanisms outperform classical ones for the error exponents in symmetric and asymmetric hypothesis testing.
As seen in Figure~\ref{Fig2}, this advantage vanishes when $\epsilon$ is greater than or equal to a threshold $\epsilon_0>0$.
Interestingly, Wang et al.\ \cite{WangSS} have recently reported a similar phenomenon.
They considered a quantum-classical setting in which the input is a bipartite pure state and the output is a classical state, and 
revealed the following phase-transition phenomenon.
Suppose that the input pure state is constrained to have entanglement entropy at least $s$.
Then there exists a threshold $s_0>0$ such that 
the optimal privacy loss $\epsilon^\star(s)$ over all such inputs remains equal to the constant $\epsilon^\star(0)$ for $s\le s_0$, 
whereas $\epsilon^\star(s)$ starts to decrease once $s$ exceeds $s_0$.
It would be interesting to investigate whether these seemingly similar threshold phenomena arise from a common underlying mathematical structure.

\section*{Acknowledgments}
Although Section~\ref{origin} originally assumed a stronger hypothesis than the saturation hypothesis, 
this part was revised thanks to a counterexample based on \cite[Corollary~3.4]{Iverson} suggested by Prof.\ Joseph W.\ Iverson, 
to whom the author extends his sincere thanks.
The author is also grateful to Prof.\ Matthew Fickus for providing a valuable reference with his comment.
Special thanks are due to Dr.\ Daiki Suruga and Dr.\ Ziyu Liu for their assistance in finding relevant references.
The author was supported by JSPS KAKENHI Grant Numbers JP22J00339 and JP22KJ1621 during the initial phase of this work.

\section*{Author contributions}
Large language models were used for some English language editing and for some literature searches.
The author independently verified references and related information identified through these searches 
before including them in the paper.
All mathematical ideas, results, and proofs presented in this paper were developed solely by the author.

\appendix
\section*{Appendices}
\section{Proposition and lemmas for isoclinic mechanisms}\label{appA}
This appendix proves that isoclinic mechanisms are $\epsilon$-QLDP.
As before, let $n,d\ge2$, $r\in[d-1]$, and $\epsilon>0$, 
representing cardinality, dimension, rank, and privacy parameter, respectively.
We begin with the following preliminary lemmas.

\begin{lemma}[Jordan's lemma {\cite[Lemma~A.1]{Nagaj}, \cite[Corollary~2.2]{Bottcher}}]\label{lemJordan}
	For all orthogonal projections $P$ and $Q$ on $\mathbb{C}^d$, 
	there exists an orthogonal decomposition 
	\begin{equation}
		\mathbb{C}^d = \bigoplus_{i,j\in\{0,1\}} \cH_{ij} \oplus \bigoplus_{k\in[m]} \cK_k
		\label{eq20}
	\end{equation}
	of $\mathbb{C}^d$ satisfying the following conditions.
	\renewcommand{\labelenumi}{$\arabic{enumi}$.}
	\begin{enumerate}
		\item
		All $\cH_{ij}$ and $\cK_k$ are invariant subspaces of both $P$ and $Q$.
		\item
		$\dim\cK_k=2$ for every $k\in[m]$.
		\item
		For each $i,j\in\{0,1\}$, 
		the restrictions of $P$ and $Q$ to $\cH_{ij}$ are equal to $iI_{ij}$ and $jI_{ij}$, respectively, 
		where $I_{ij}$ is the identity operator on $\cH_{ij}$.
		\item
		For every $k\in[m]$, the restrictions of $P$ and $Q$ to $\cK_k$ are rank-one orthogonal projections.
	\end{enumerate}
\end{lemma}

\begin{lemma}\label{lem1}
	Assume that an orthogonal decomposition of $\mathbb{C}^d$ 
	with respect to orthogonal projections $P$ and $Q$ on $\mathbb{C}^d$ is given as \eqref{eq20}.
	If $\cH_{01}=\cH_{10}=\cH_{11}=\{0\}$, 
	then the maximum (resp.\ minimum) eigenvalue $\lambda_{+}$ (resp.\ $\lambda_{-}$) of $e^\epsilon P-Q$ is 
	\[
	\lambda_\pm = e^{\epsilon/2}\Bigl( \sinh(\epsilon/2) \pm \max_{k\in[m]} \sqrt{\sinh^2(\epsilon/2)+1-c_k} \Bigr),
	\]
	where $c_k$ is the trace of the product of the restrictions of $P$ and $Q$ to $\cK_k$: 
	$c_k=\Tr(P|_{\cK_k})(Q|_{\cK_k})$.
	Moreover, we have $\sum_{k\in[m]} c_k = \Tr PQ$.
\end{lemma}
\begin{proof}
	By Lemma~\ref{lemJordan} and the assumption $\cH_{01}=\cH_{10}=\cH_{11}=\{0\}$, 
	the eigenvalue $\lambda_{+}$ (resp.\ $\lambda_{-}$) of $e^\epsilon P-Q$ is equal to 
	the maximum (resp.\ minimum) of the maximum (resp.\ minimum) eigenvalues of $e^\epsilon P|_{\cK_k}-Q|_{\cK_k}$, $k\in[m]$.
	Since $P|_{\cK_k}$ and $Q|_{\cK_k}$ are rank-one orthogonal projections, 
	they can be expressed as 
	\[
	U_k(P|_{\cK_k})U_k^\dag = 
	\begin{bmatrix}
		1&0\\
		0&0
	\end{bmatrix}
	\quad\text{and}\quad
	U_k(Q|_{\cK_k})U_k^\dag = 
	\begin{bmatrix}
		c_k&\sqrt{c_k(1-c_k)}\\
		\sqrt{c_k(1-c_k)}&1-c_k
	\end{bmatrix}
	\]
	by some unitary matrix $U_k$.
	Thus, the determinant $D_k$ and the trace $T_k$ of $e^\epsilon P|_{\cK_k}-Q|_{\cK_k}$ are 
	\[
	D_k = -e^\epsilon(1-c_k)
	\quad\text{and}\quad
	T_k = e^\epsilon-1 = 2e^{\epsilon/2}\sinh(\epsilon/2),
	\]
	respectively. Therefore, 
	\begin{align*}
		\lambda_\pm &= T_k/2 \pm \max_{k\in[m]} \sqrt{(T_k/2)^2-D_k}\\
		&= e^{\epsilon/2}\Bigl( \sinh(\epsilon/2) \pm \max_{k\in[m]} \sqrt{\sinh^2(\epsilon/2)+1-c_k} \Bigr).
	\end{align*}
\end{proof}

\begin{lemma}\label{lem2}
	Assume that an orthogonal decomposition of $\mathbb{C}^d$ 
	with respect to orthogonal projections $P$ and $Q$ on $\mathbb{C}^d$ is given as \eqref{eq20}; 
	assume that $PQP=cP$ and $QPQ=cQ$ for some $0<c<1$.
	Then $\cH_{01}=\cH_{10}=\cH_{11}=\{0\}$ and $\Tr(P|_{\cK_k})(Q|_{\cK_k})=c$ for every $k\in[m]$.
\end{lemma}
\begin{proof}
	We show $\cH_{01}=\cH_{10}=\cH_{11}=\{0\}$.
	Let $(i,j)\in\{0,1\}$ and $u\in\cH_{ij}\setminus\{0\}$.
	By $i^2ju=PQPu=cPu=ciu$ and $ij^2u=QPQu=cQu=cju$, 
	we have $i(c-ij)=0$ and $j(c-ij)=0$.
	These and the assumption $0<c<1$ yield $i=j=0$.
	Therefore, $\cH_{01}=\cH_{10}=\cH_{11}=\{0\}$.
	\par
	Let $k\in[m]$.
	We show $\Tr(P|_{\cK_k})(Q|_{\cK_k})=c$.
	Since $\cK_k$ is an invariant subspace of both $P$ and $Q$, 
	the assumption $PQP=cP$ implies $(P|_{\cK_k})(Q|_{\cK_k})(P|_{\cK_k})=cP|_{\cK_k}$.
	Since $P|_{\cK_k}$ is a rank-one orthogonal projection, 
	it follows that $\Tr(P|_{\cK_k})(Q|_{\cK_k})=c$.
\end{proof}

\begin{lemma}\label{lem3}
	Let $\{P_i\}_{i\in[n]}$ be an EITFF${}_{\mathbb{C}}(d,r,n)$ satisfying $n>d/r$.
	Then, for all distinct integers $i,j\in[n]$, 
	the maximum (resp.\ minimum) eigenvalue of $e^\epsilon P_i-P_j$ is equal to 
	the following $\lambda_{+}$ (resp.\ $\lambda_{-}$): 
	\[
	\lambda_\pm = e^{\epsilon/2}\Bigl( \sinh(\epsilon/2) \pm \sqrt{\sinh^2(\epsilon/2)+1-c} \Bigr),
	\]
	where $c=\frac{nr-d}{d(n-1)}$.
\end{lemma}
\begin{proof}
	By $c=\frac{nr-d}{d(n-1)}$, $n>d/r$, and $r\in[d-1]$, we have $0<c<1$.
	Thus, the assertion follows from Lemmas~\ref{lem1} and \ref{lem2}.
\end{proof}

We now prove that isoclinic mechanisms are $\epsilon$-QLDP, using Lemma~\ref{lem3}.

\begin{proposition}\label{prop1}
	Let $\{P_i\}_{i\in[n]}$ be an EITFF${}_{\mathbb{C}}(d,r,n)$, and 
	$\mu$ be a real number in the interval $[0,d/(d-r)]$.
	For $x\in[n]$, set the density matrix $\sigma_x$ as \eqref{eq13} regardless of \eqref{eq14}.
	Then the mechanism $\vec{\sigma}=(\sigma_1,\ldots,\sigma_n)$ is $\epsilon$-QLDP  
	if and only if  
	\begin{equation}
		\frac{1}{1-h_{+}} \le 1-\mu \le \frac{1}{1-h_{-}},
		\label{eq21}
	\end{equation}
	where 
	\[
	h_\pm = \frac{d}{2r}\biggl( 1\pm\sqrt{1+\frac{1-c}{\sinh^2(\epsilon/2)}} \biggr)
	\quad\text{and}\quad
	c=\frac{nr-d}{d(n-1)}.
	\]
	In particular, the isoclinic mechanism $\vec{\sigma}$ in \eqref{eq13} is $\epsilon$-QLDP.
\end{proposition}
\begin{proof}
	First, assume $n=d/r$.
	Then $c=0$ and $h_\pm = (n/2)( 1\pm\coth(\epsilon/2) )$.
	Thus, \eqref{eq21} is equivalent to 
	\begin{equation}
		\frac{n}{n-1+e^\epsilon} \le \mu \le \frac{n}{n-1+e^{-\epsilon}}.
		\label{eq22}
	\end{equation}
	Since $\sigma_1,\ldots,\sigma_n$ are commutative, we have the following equivalences: 
	\begin{align*}
		\text{$\vec{\sigma}$ is $\epsilon$-QLDP}
		&\iff \frac{\mu}{d}e^{-\epsilon} \le \frac{\mu}{d}+\frac{1-\mu}{r} \le \frac{\mu}{d}e^\epsilon\\
		&\iff \mu e^{-\epsilon} \le \mu+n(1-\mu) \le \mu e^\epsilon
		\iff \text{\eqref{eq22}}.
	\end{align*}
	Therefore, the desired equivalence holds.
	\par
	Next, assume $n>d/r$.
	By Lemma~\ref{lem3}, the mechanism $\vec{\sigma}$ is $\epsilon$-QLDP 
	if and only if  
	\begin{equation}
		(e^\epsilon-1)\frac{\mu}{d} + \frac{1-\mu}{r}\lambda_\pm \ge 0
		\label{eq23}
	\end{equation}
	for both signs, where $\lambda_\pm$ are defined in Lemma~\ref{lem3}.
	Note that \eqref{eq23} implies $\mu\in[0,d/(d-r)]$ because $\lambda_{-}<0$ and $\lambda_{+}>e^\epsilon-1$.
	Also, for both signs, 
	\[
	h_\pm = \frac{d}{r}\cdot\frac{\lambda_\pm}{e^\epsilon-1}.
	\]
	Considering the two cases $\mu\le1$ and $\mu>1$, we obtain the desired equivalence.
\end{proof}

\section{Proofs of Theorems~$\ref{main1}$, $\ref{main2}$, and $\ref{main3}$}\label{appB}
This appendix proves Theorems~\ref{main1}, \ref{main2}, and \ref{main3}.
As before, let $n,d\ge2$, $r\in[d-1]$, and $\epsilon>0$, 
representing cardinality, dimension, rank, and privacy parameter, respectively.
We use the following technical lemma \cite[Theorem~4]{Kairouz}.

\begin{lemma}\label{lemKairouz}
	The classical optimal value $\OPT_n(\epsilon; \Phi_\C)$ is equal to 
	the maximum of the sum 
	\begin{equation}
		\sum_{\mathbf{z}\in\{0,1\}^n} \phi(\mathbbm{1}_n+(e^\epsilon-1)\mathbf{z})\alpha_{\mathbf{z}}
		\label{eq24}
	\end{equation}
	subject to the following constraints.
	\renewcommand{\labelenumi}{$\arabic{enumi}$.}
	\begin{enumerate}
		\item
		$\alpha_{\mathbf{z}}\ge0$ for every $\mathbf{z}\in\{0,1\}^n$.
		\item
		$\sum_{\mathbf{z}\in\{0,1\}^n} \alpha_{\mathbf{z}}(\mathbbm{1}_n+(e^\epsilon-1)\mathbf{z})=\mathbbm{1}_n$.
	\end{enumerate}
\end{lemma}

We now prove Theorems~\ref{main1}, \ref{main2}, and \ref{main3}.
Assume all the assumptions of Theorem~\ref{main1} here and below.
We first prove the upper bounds for the classical and quantum optimal values, namely, 
\begin{align}
	\OPT_n(\epsilon; \Phi_\C) &\le \phi(\mathbbm{1}_n) + \frac{\lfloor{n/2}\rfloor\lceil{n/2}\rceil}{n-1}\cdot\frac{\beta_0}{2}\epsilon^2 + o(\epsilon^2)
	\quad\text{and}\quad\label{eq-opt-c-ub}\\
	\OPT_n(\epsilon; \Phi_\Q) &\le \phi(\mathbbm{1}_n) + \frac{\beta_0n}{4}\epsilon^2 + o(\epsilon^2),
	\label{eq-opt-q-ub}
\end{align}
and then prove the lower bounds for them, namely, 
\begin{align}
	\OPT_n(\epsilon; \Phi_\C) &\ge \Phi_\C(b^*) \ge \phi(\mathbbm{1}_n) + \frac{\lfloor{n/2}\rfloor\lceil{n/2}\rceil}{n-1}\cdot\frac{\beta_0}{2}\epsilon^2 + o(\epsilon^2)
	\quad\text{and}\quad\label{eq-opt-c-lb}\\
	\OPT_n(\epsilon; \Phi_\Q) &\ge \Phi_\Q(\vec{\sigma}^*) \ge \phi(\mathbbm{1}_n) + \frac{\beta_0n}{4}\epsilon^2 + o(\epsilon^2).
	\label{eq-opt-q-lb}
\end{align}
These inequalities imply Theorems~\ref{main1}, \ref{main2}, and \ref{main3}.

\begin{proof}[Proofs of $\eqref{eq-opt-c-ub}$ and $\eqref{eq-opt-q-ub}$]
	By assumption~1 of Theorem~\ref{main1}, 
	the Hesse matrix $\mathbf{H}_\phi(\mathbbm{1}_n)=[\partial_i\partial_j\phi(\mathbbm{1}_n)]_{i,j\in[n]}$ can be expressed as 
	$\mathbf{H}_\phi(\mathbbm{1}_n)=\beta_1I_n+\beta_2\mathbbm{1}_n\mathbbm{1}_n^\top$ 
	for some real numbers $\beta_1$ and $\beta_2$.
	Euler's homogeneous function theorem implies 
	\begin{equation}
		\sum_{j\in[n]} z_j\partial_j\phi(\mathbf{z})=\phi(\mathbf{z}).
		\label{eq25}
	\end{equation}
	Differentiating both sides partially with respect to $z_i$, and 
	substituting $\mathbf{z}=\mathbbm{1}_n$, 
	we obtain $\sum_{j\in[n]} \partial_i\partial_j\phi(\mathbbm{1}_n)=0$ for every $i\in[n]$, 
	which implies $0=\beta_1+n\beta_2$.
	Thus, the Hesse matrix $\mathbf{H}_\phi(\mathbbm{1}_n)$ is a positive constant multiple of $nI_n-\mathbbm{1}_n\mathbbm{1}_n^\top$: 
	for all $i,j\in[n]$, 
	\begin{equation}
		\partial_i\partial_j\phi(\mathbbm{1}_n)=\beta_0\gamma_{i,j}
		\quad\text{and}\quad
		\gamma_{i,j} = 
		\begin{cases}
			1 & \text{if }i=j,\\
			-1/(n-1) & \text{if }i\not=j,
		\end{cases}
		\label{eq26}
	\end{equation}
	where $\beta_0=\partial_1^2\phi(\mathbbm{1}_n)>0$ by assumption~2 of Theorem~\ref{main1}.
	\par
	\textbf{Proof of $\eqref{eq-opt-c-ub}$.}
	By Lemma~\ref{lemKairouz}, it suffices to show that 
	\begin{equation}
		\text{\eqref{eq24}}
		\le \phi(\mathbbm{1}_n) + \frac{\lfloor{n/2}\rfloor\lceil{n/2}\rceil}{n-1}\cdot\frac{\beta_0}{2}\epsilon^2 + o(\epsilon^2)
		\label{eq27}
	\end{equation}
	as $\epsilon\to+0$, subject to constraints~1--2 in Lemma~\ref{lemKairouz}.
	The second-order Taylor expansion of $\phi(\mathbbm{1}_n+(e^\epsilon-1)\mathbf{z})$ is 
	\begin{align*}
		\phi(\mathbbm{1}_n+(e^\epsilon-1)\mathbf{z})
		&= \phi(\mathbbm{1}_n) + (e^\epsilon-1)\sum_{i\in[n]} \partial_i\phi(\mathbbm{1}_n)z_i\\
		&\quad+ \frac{(e^\epsilon-1)^2}{2}\sum_{i,j\in[n]} \partial_i\partial_j\phi(\mathbbm{1}_n)z_iz_j
		+ o( (e^\epsilon-1)^2 ).
	\end{align*}
	Thus, 
	\begin{equation}
	\begin{split}
		\text{\eqref{eq24}} &= \phi(\mathbbm{1}_n)\sum_{\mathbf{z}\in\{0,1\}^n} \alpha_{\mathbf{z}}
		+ (e^\epsilon-1)\sum_{\mathbf{z}\in\{0,1\}^n} \sum_{i\in[n]} \partial_i\phi(\mathbbm{1}_n)z_i\alpha_{\mathbf{z}}\\
		&\quad+ \frac{(e^\epsilon-1)^2}{2}\sum_{\mathbf{z}\in\{0,1\}^n} \sum_{i,j\in[n]} \partial_i\partial_j\phi(\mathbbm{1}_n)z_iz_j\alpha_{\mathbf{z}}
		+ o( (e^\epsilon-1)^2 ).
	\end{split}\label{eq28}
	\end{equation}
	Substituting $\mathbf{z}=\mathbbm{1}_n$ into \eqref{eq25}, we have 
	\begin{equation}
		\sum_{i\in[n]} \partial_i\phi(\mathbbm{1}_n) = \phi(\mathbbm{1}_n).
		\label{eq29}
	\end{equation}
	By this and constraint~2 in Lemma~\ref{lemKairouz}, the linear term of \eqref{eq28} is 
	\begin{align*}
		(e^\epsilon-1)\sum_{\mathbf{z}\in\{0,1\}^n} \sum_{i\in[n]} \partial_i\phi(\mathbbm{1}_n)z_i\alpha_{\mathbf{z}}
		&= \sum_{i\in[n]} \partial_i\phi(\mathbbm{1}_n)\biggl( 1 - \sum_{\mathbf{z}\in\{0,1\}^n} \alpha_{\mathbf{z}} \biggr)\\
		&= \phi(\mathbbm{1}_n)\biggl( 1 - \sum_{\mathbf{z}\in\{0,1\}^n} \alpha_{\mathbf{z}} \biggr).
	\end{align*}
	This, \eqref{eq28}, and \eqref{eq26} yield 
	\begin{align}
		\text{\eqref{eq24}}
		&= \phi(\mathbbm{1}_n)
		+ \frac{(e^\epsilon-1)^2}{2}\sum_{\mathbf{z}\in\{0,1\}^n} \sum_{i,j\in[n]} \partial_i\partial_j\phi(\mathbbm{1}_n)z_iz_j\alpha_{\mathbf{z}}
		+ o( (e^\epsilon-1)^2 )\nonumber\\
		&= \phi(\mathbbm{1}_n)
		+ \frac{(e^\epsilon-1)^2}{2}\sum_{\mathbf{z}\in\{0,1\}^n} \sum_{i,j\in[n]} \beta_0\gamma_{i,j}z_iz_j\alpha_{\mathbf{z}}
		+ o( (e^\epsilon-1)^2 ). \label{eq30}
	\end{align}
	\par
	Let us bound the sum 
	\begin{equation}
		\sum_{\mathbf{z}\in\{0,1\}^n} \sum_{i,j\in[n]} \gamma_{i,j}z_iz_j\alpha_{\mathbf{z}}
		= \sum_{k=0}^n \sum_{\substack{\mathbf{z}\in\{0,1\}^n \\ \|\mathbf{z}\|_1=k}} \sum_{i,j\in[n]} \gamma_{i,j}z_iz_j\alpha_{\mathbf{z}}
		\label{eq31}
	\end{equation}
	from above, subject to constraints~1--2 in Lemma~\ref{lemKairouz}.
	If $\|\mathbf{z}\|_1=k$ and $\mathbf{z}\in\{0,1\}^n$, then 
	\[
	\sum_{i,j\in[n]} \gamma_{i,j}z_iz_j
	= k - \frac{k(k-1)}{n-1}
	= \frac{k(n-k)}{n-1},
	\]
	which achieves the maximum at $k\in\{ \lfloor{n/2}\rfloor,\lceil{n/2}\rceil \}$.
	This and constraints~1--2 in Lemma~\ref{lemKairouz} imply 
	\[
	\text{\eqref{eq31}} \le \frac{\lfloor{n/2}\rfloor\lceil{n/2}\rceil}{n-1}\sum_{\mathbf{z}\in\{0,1\}^n} \alpha_{\mathbf{z}}
	\le \frac{\lfloor{n/2}\rfloor\lceil{n/2}\rceil}{n-1}.
	\]
	By this, \eqref{eq30}, and assumption~2 of Theorem~\ref{main1}, 
	\[
	\text{\eqref{eq24}}
	\le \phi(\mathbbm{1}_n)
	+ \frac{\beta_0}{2}\cdot\frac{\lfloor{n/2}\rfloor\lceil{n/2}\rceil}{n-1}(e^\epsilon-1)^2
	+ o( (e^\epsilon-1)^2 ),
	\]
	which shows \eqref{eq27}.
	\par
	\textbf{Proof of $\eqref{eq-opt-q-ub}$.}
	Recall that $\rho_{\avg}$ is defined as $\rho_{\avg}=(1/n)\sum_{x\in[n]} \rho_x$.
	Set $\delta\rho_x=\rho_x-\rho_{\avg}$ for $x\in[n]$.
	By assumption~4 of Theorem~\ref{main1}, the function $\Phi_\Q(\vec{\rho})$ has the second-order Taylor expansion around $(\rho_{\avg},\ldots,\rho_{\avg})$: 
	\begin{gather}
		\Phi_\Q(\vec{\rho}) = \phi(\mathbbm{1}_n) + \Phi_\Q^{(1)}(\rho_{\avg}; \delta\vec{\rho}) + \Phi_\Q^{(2)}(\rho_{\avg}; \delta\vec{\rho})
		+ o\Bigl( \max_{i\in[n]} \|\delta\rho_i\|^2 \Bigr), \label{eq67}\\
		\Phi_\Q^{(1)}(\rho_{\avg}; \delta\vec{\rho}) = \sum_{i\in[n]} \cA_{i;\rho_{\avg}}[\delta\rho_i], \quad\text{and}\quad
		\Phi_\Q^{(2)}(\rho_{\avg}; \delta\vec{\rho}) = \frac{\beta_0n}{2(n-1)}\sum_{i\in[n]} J_{\rho_{\avg}}[\delta\rho_i,\delta\rho_i], \nonumber
	\end{gather}
	where $\cA_{i;\rho_{\avg}}$ is the first-order partial Fr\'echet derivative 
	at $(\rho_{\avg},\ldots,\rho_{\avg})$ with respect to the $i$th variable, and 
	we have used Lemma~\ref{lemFormulation} to obtain the constant term.
	By assumption~3 of Theorem~\ref{main1}, $\cA_{i;\rho_{\avg}}=\cA_{j;\rho_{\avg}}$ for all $i,j\in[n]$.
	Thus, the linear term vanishes: 
	\begin{equation}
		\Phi_\Q^{(1)}(\rho_{\avg}; \delta\vec{\rho})
		= \cA_{1;\rho_{\avg}}\biggl[ \sum_{i\in[n]} \delta\rho_i \biggr]
		= 0. \label{eq32}
	\end{equation}
	Also, it can easily be checked that $\|\delta\rho_x\|=O(\epsilon)$ as $\epsilon\to+0$, for every $x\in[n]$.
	\par
	We show \eqref{eq-opt-q-ub}.
	By \eqref{eq32}, it suffices to show that 
	\begin{equation}
		\Phi_\Q^{(2)}(\rho_{\avg}; \delta\vec{\rho})
		\le \frac{\beta_0n}{4}\epsilon^2 + o(\epsilon^2)
		\label{eq33}
	\end{equation}
	as $\epsilon\to+0$.
	Set $\Delta_x=(\rho_{\avg})^{-1/2}\delta\rho_x(\rho_{\avg})^{-1/2}$ for $x\in[n]$.
	By \eqref{eq-SLD-RLD} and assumption~2 of Theorem~\ref{main1}, 
	\begin{equation}
	\begin{split}
		\Phi_\Q^{(2)}(\rho_{\avg}; \delta\vec{\rho})
		&\le \frac{\beta_0n}{2(n-1)}\sum_{i\in[n]} J_{\rho_{\avg}}^{\RLD}[\delta\rho_i,\delta\rho_i]
		= \frac{\beta_0n}{2(n-1)}\sum_{i\in[n]} \Tr(\delta\rho_i)^2(\rho_{\avg})^{-1}\\
		&= \frac{\beta_0n}{2(n-1)}\sum_{i\in[n]} \Tr\Delta_i^2\rho_{\avg}
		= \frac{\beta_0}{4(n-1)}\sum_{i,j\in[n]} \Tr(\Delta_i-\Delta_j)^2\rho_{\avg},
	\end{split}\label{eq34}
	\end{equation}
	where we have used $\sum_{i\in[n]} \Delta_i=0$ to obtain the last equality.
	By \eqref{eq01}, 
	the following inequalities hold: 
	\begin{gather*}
		e^{-\epsilon}\rho_{\avg}\le\rho_i\le e^\epsilon\rho_{\avg},\quad
		(e^{-\epsilon}-1)I_d \le \Delta_i \le (e^\epsilon-1)I_d,
		\quad\text{and}\\
		(e^{-\epsilon}-1)(I_d+\Delta_j) \le \Delta_i-\Delta_j \le (e^\epsilon-1)(I_d+\Delta_j)
	\end{gather*}
	for all $i,j\in[n]$.
	Since these yield that $\|\Delta_i\|\le\epsilon+o(\epsilon)$ and $\|\Delta_i-\Delta_j\|\le\epsilon+o(\epsilon)$, 
	it follows that 
	\[
	\Tr(\Delta_i-\Delta_j)^2\rho_{\avg}
	\le \|(\Delta_i-\Delta_j)^2\|\,\|\rho_{\avg}\|_1
	\le \|\Delta_i-\Delta_j\|^2
	\le \epsilon^2+o(\epsilon^2)
	\]
	for all $i,j\in[n]$.
	Therefore, 
	\[
	\sum_{i,j\in[n]} \Tr(\Delta_i-\Delta_j)^2\rho_{\avg}
	\le \sum_{i\not=j} (\epsilon^2+o(\epsilon^2))
	= n(n-1)\epsilon^2+o(\epsilon^2)
	\]
	for all $i,j\in[n]$.
	This, \eqref{eq34}, and assumption~2 of Theorem~\ref{main1} yield \eqref{eq33}.
\end{proof}

\begin{proof}[Proof of $\eqref{eq-opt-q-lb}$]
	For simplicity, write $\sigma_x^*$ and $\mu_*$ as $\sigma_x$ and $\mu$, respectively, in the following calculation.
	Set $\sigma_{\avg}=(1/n)\sum_{k\in[n]} \sigma_k=d^{-1}I_d$, $\nu=1-\mu$, and 
	$\delta\sigma_x=\sigma_x-\sigma_{\avg}=\nu(r^{-1}P_x-d^{-1}I_d)$ for $x\in[n]$.
	By assumption~4 of Theorem~\ref{main1} and \eqref{eq-SLD-RLD}, 
	\[
	\Phi_\Q^{(2)}(\sigma_{\avg}; \delta\vec{\sigma})
	\ge \frac{\beta_0n}{2(n-1)}\sum_{i\in[n]} J_{\sigma_{\avg}}^{\SLD}[\delta\sigma_i,\delta\sigma_i].
	\]
	Each $J_{\sigma_{\avg}}^{\SLD}[\delta\sigma_i,\delta\sigma_i]$ can be calculated as follows: 
	\begin{align*}
		J_{\sigma_{\avg}}^{\SLD}[\delta\sigma_i,\delta\sigma_i]
		&= \int_0^\infty \Tr e^{-(t/2)\sigma_{\avg}}\delta\sigma_ie^{-(t/2)\sigma_{\avg}}\delta\sigma_i\,\mathrm{d}t\\
		&= \biggl( \int_0^\infty e^{-t/d}\,\mathrm{d}t \biggr)\Tr(\delta\sigma_i)^2
		= d\Tr(\delta\sigma_i)^2\\
		&= d\nu^2\Tr(r^{-1}P_i-d^{-1}I_d)^2
		= d\nu^2(r^{-1}-d^{-1}).
	\end{align*}
	Thus, 
	\begin{equation}
		\Phi_\Q^{(2)}(\sigma_{\avg}; \delta\vec{\sigma})
		\ge \frac{\beta_0n}{2(n-1)}\sum_{i\in[n]} d\nu^2(r^{-1}-d^{-1})
		= \frac{\beta_0n^2}{2(n-1)}\cdot\frac{d-r}{r}\nu^2.
		\label{eq35}
	\end{equation}
	By \eqref{eq15}, 
	\begin{equation*}
		\nu^2 = \biggl( 1+\frac{1-c}{\sinh^2(\epsilon/2)} \biggr)^{-1}
		= \frac{\epsilon^2}{4(1-c)}+o(\epsilon^2)
	\end{equation*}
	as $\epsilon\to+0$.
	By \eqref{eq35}, $d=2r$, and $1-c=\frac{n}{2n-2}$, 
	\begin{equation*}
	\begin{split}
		\Phi_\Q^{(2)}(\sigma_{\avg}; \delta\vec{\sigma})
		&\ge \frac{\beta_0n^2}{2(n-1)}\nu^2
		= \frac{\beta_0n^2}{2(n-1)}\cdot\frac{\epsilon^2}{4(1-c)} + o(\epsilon^2)\\
		&= \frac{\beta_0n^2}{2(n-1)}\cdot\frac{\epsilon^2}{4}\cdot\frac{2n-2}{n} + o(\epsilon^2)
		= \frac{\beta_0n}{4}\epsilon^2 + o(\epsilon^2)
	\end{split}
	\end{equation*}
	as $\epsilon\to+0$.
	Substituting $\vec{\rho}=\vec{\sigma}$ into \eqref{eq67} and \eqref{eq32}, 
	we obtain 
	\begin{equation*}
		\Phi_\Q(\vec{\sigma}) = \phi(\mathbbm{1}_n) + \Phi_\Q^{(2)}(\sigma_{\avg}; \delta\vec{\sigma}) + o(\epsilon^2).
	\end{equation*}
	Therefore, the second inequality of \eqref{eq-opt-q-lb} holds.
	The first inequality of \eqref{eq-opt-q-lb} is clear.
\end{proof}

\begin{proof}[Proof of $\eqref{eq-opt-c-lb}$]
	Set $k=\lfloor{n/2}\rfloor$.
	By assumptions~1--2 of Theorem~\ref{main1} and Taylor's theorem, 
	\begin{align*}
		\Phi_\C(b^*) &= \phi( b^*(1|1),\ldots,b^*(1|n) ) + \phi( b^*(2|1),\ldots,b^*(2|n) )\\
		&= \frac{1}{e^\epsilon+1}\biggl( \phi(\underbrace{e^\epsilon,\ldots,e^\epsilon}_{k}, \underbrace{1,\ldots,1}_{n-k})
		+ \phi(\underbrace{1,\ldots,1}_{k}, \underbrace{e^\epsilon,\ldots,e^\epsilon}_{n-k}) \biggr)\\
		&= \frac{1}{e^\epsilon+1}\biggl( 2\phi(\mathbbm{1}_n) + \sum_{i\in[n]} \partial_i\phi(\mathbbm{1}_n)(e^\epsilon-1)\\
		&\quad+ \frac{1}{2}\sum_{i,j\in[k]} \partial_i\partial_j\phi(\mathbbm{1}_n)(e^\epsilon-1)^2
		+ \frac{1}{2}\sum_{i,j\in[n]\setminus[k]} \partial_i\partial_j\phi(\mathbbm{1}_n)(e^\epsilon-1)^2 \biggr)\\
		&\quad+ o( (e^\epsilon-1)^2 ).
	\end{align*}
	By \eqref{eq29} and \eqref{eq26}, 
	\begin{align*}
		\Phi_\C(b^*)
		&= \frac{1}{e^\epsilon+1}\biggl( 2\phi(\mathbbm{1}_n) + \phi(\mathbbm{1}_n)(e^\epsilon-1)\\
		&\quad+ \frac{1}{2}\sum_{i,j\in[k]} \beta_0\gamma_{i,j}(e^\epsilon-1)^2
		+ \frac{1}{2}\sum_{i,j\in[n]\setminus[k]} \beta_0\gamma_{i,j}(e^\epsilon-1)^2 \biggr)\\
		&\quad+ o( (e^\epsilon-1)^2 )\\
		&= \phi(\mathbbm{1}_n) + \frac{\beta_0(e^\epsilon-1)^2}{e^\epsilon+1}\cdot\frac{k(n-k)}{n-1} + o( (e^\epsilon-1)^2 )\\
		&= \phi(\mathbbm{1}_n) + \frac{\beta_0\epsilon^2}{2}\cdot\frac{k(n-k)}{n-1} + o( \epsilon^2 )
	\end{align*}
	as $\epsilon\to+0$, which shows the second inequality of \eqref{eq-opt-c-lb}.
	The first inequality of \eqref{eq-opt-c-lb} is clear.
\end{proof}

\section{Proofs of Theorems~$\ref{main1'}$, $\ref{main2'}$, and $\ref{main3'}$}\label{appC}
This appendix proves Theorems~\ref{main1'}, \ref{main2'}, and \ref{main3'}.
As before, let $n,d\ge2$ and $\epsilon>0$, 
representing cardinality, dimension, and privacy parameter, respectively.
We begin with the following preliminary lemmas.

\begin{lemma}[cf.\ {\cite[Theorem~4]{Grace}}]\label{lemChernoff}
	Let $\rho_0,\rho_1,\rho_2$ be full-rank quantum states on $\mathbb{C}^d$.
	Set $\delta\rho_k=\rho_k-\rho_0$ for $k\in[2]$.
	Then the Chernoff information $C(\rho_1,\rho_2)$ has the second-order Taylor expansion around $(\rho_0,\rho_0)$: 
	\[
	C(\rho_1,\rho_2)
	= \frac{1}{8}J_{\rho_0}^{\WYD(1/2)}[\delta\rho_1-\delta\rho_2,\delta\rho_1-\delta\rho_2]
	+ o( \max\{ \|\delta\rho_1\|^2,\|\delta\rho_2\|^2 \} )
	\]
	as $\|\delta\rho_1\|,\|\delta\rho_2\|\to0$.
\end{lemma}
\begin{proof}
	By \cite[p.~4]{Audenaert}, \cite[Eq.~(46)]{Calsamiglia}, 
	the Chernoff information $C(\rho_1,\rho_2)$ has the second-order Taylor expansion around $(\rho_1,\rho_1)$: 
	\begin{equation}
		C(\rho_1,\rho_2)
		= \frac{1}{8}J_{\rho_1}^{\WYD(1/2)}[\rho_2-\rho_1,\rho_2-\rho_1]
		+ o( \|\rho_2-\rho_1\|^2 ).
		\label{eq40}
	\end{equation}
	Since the map $\rho\mapsto J_{\rho}^{\WYD(1/2)}$ is smooth, 
	it is Lipschitz continuous on every compact neighborhood of $\rho_0$.
	Thus, there exist positive numbers $C$ and $\delta_0$ such that 
	\[
	\abs{J_{\rho_1}^{\WYD(1/2)}[X,X] - J_{\rho_0}^{\WYD(1/2)}[X,X]}
	\le C\,\|\delta\rho_1\|\,\|X\|^2
	\]
	for every Hermitian matrix $X$ and every traceless Hermitian matrix $\delta\rho_1$ satisfying $\|\delta\rho_1\|\le\delta_0$.
	This, $\rho_2-\rho_1=\delta\rho_2-\delta\rho_1$, and \eqref{eq40} yield that 
	\begin{align*}
		C(\rho_1,\rho_2)
		&= \frac{1}{8}J_{\rho_0}^{\WYD(1/2)}[\rho_2-\rho_1,\rho_2-\rho_1]
		+ O(\|\delta\rho_1\|\,\|\rho_2-\rho_1\|^2)
		+ o( \|\rho_2-\rho_1\|^2 )\\
		&= \frac{1}{8}J_{\rho_0}^{\WYD(1/2)}[\delta\rho_2-\delta\rho_1,\delta\rho_2-\delta\rho_1]
		+ o( \max\{ \|\delta\rho_1\|^2,\|\delta\rho_2\|^2 \} )
	\end{align*}
	as $\|\delta\rho_1\|,\|\delta\rho_2\|\to0$.
\end{proof}

\begin{lemma}[cf.\ {\cite[Theorem~4]{Grace}}]\label{lemTr}
	Let $\rho_0,\rho_1,\rho_2$ be full-rank quantum states on $\mathbb{C}^d$.
	Set $\delta\rho_k=\rho_k-\rho_0$ for $k\in[2]$.
	Then, for every $s\in(0,1)$, 
	the function $\Tr\rho_1^s\rho_2^{1-s}$ has the second-order Taylor expansion around $(\rho_0,\rho_0)$: 
	\[
	\Tr\rho_1^s\rho_2^{1-s}
	= 1 - \frac{s(1-s)}{2}J_{\rho_0}^{\WYD(s)}[\delta\rho_1-\delta\rho_2,\delta\rho_1-\delta\rho_2]
	+ o( \max\{ \|\delta\rho_1\|^2,\|\delta\rho_2\|^2 \} )
	\]
	as $\|\delta\rho_1\|,\|\delta\rho_2\|\to0$.
\end{lemma}
\begin{proof}
	Let $s\in(0,1)$.
	By \cite[Eq.~(46)]{Calsamiglia}, 
	the function $-\ln\Tr\rho_1^s\rho_2^{1-s}$ has the second-order Taylor expansion around $(\rho_1,\rho_1)$: 
	\[
	-\ln\Tr\rho_1^s\rho_2^{1-s}
	= \frac{s(1-s)}{2}J_{\rho_1}^{\WYD(s)}[\rho_2-\rho_1,\rho_2-\rho_1]
	+ o(\|\rho_2-\rho_1\|^2).
	\]
	Thus, 
	\begin{align*}
		\Tr\rho_1^s\rho_2^{1-s}
		&= \exp(\ln\Tr\rho_1^s\rho_2^{1-s})
		= 1 + \ln\Tr\rho_1^s\rho_2^{1-s} + o(\|\rho_2-\rho_1\|^2)\\
		&= 1 - \frac{s(1-s)}{2}J_{\rho_1}^{\WYD(s)}[\rho_2-\rho_1,\rho_2-\rho_1] + o(\|\rho_2-\rho_1\|^2).
	\end{align*}
	The remainder of the proof is the same as the proof of Lemma~\ref{lemChernoff}.
\end{proof}

\begin{lemma}\label{lemQLDP}
	Let $\vec{\rho}$ be an $\epsilon$-QLDP mechanism, and let $\eta\in(0,1]$.
	For $k\in[n]$, define $\tilde{\rho}_k$ as \eqref{eq11}, i.e., 
	$\tilde{\rho}_k=\eta\rho_k+(1-\eta)\rho_{\avg}$.
	Then the collection of $\tilde{\rho}_1,\ldots,\tilde{\rho}_n$ is $\eta\epsilon(1+\sqrt{\epsilon})$-QLDP as $\epsilon\to+0$.
\end{lemma}
\begin{proof}
	Since the collection of $\tilde{\rho}_1,\ldots,\tilde{\rho}_n$ is $\epsilon$-QLDP, 
	the assertion is clear whenever $\eta\ge(1+\sqrt{\epsilon})^{-1}$.
	Assume $\eta<(1+\sqrt{\epsilon})^{-1}$ below.
	By \eqref{eq01}, 
	\begin{align*}
		\tilde{\rho}_i
		&= \eta e^{\eta\epsilon(1+\sqrt{\epsilon})-\epsilon}\rho_i + \eta(1-e^{\eta\epsilon(1+\sqrt{\epsilon})-\epsilon})\rho_i + (1-\eta)\rho_{\avg}\\
		&\le \eta e^{\eta\epsilon(1+\sqrt{\epsilon})}\rho_j + \eta(e^\epsilon-e^{\eta\epsilon(1+\sqrt{\epsilon})})\rho_{\avg} + (1-\eta)\rho_{\avg}.
	\end{align*}
	Thus, it suffices to show that 
	$\eta(e^\epsilon-e^{\eta\epsilon(1+\sqrt{\epsilon})}) < (1-\eta)(e^{\eta\epsilon(1+\sqrt{\epsilon})}-1)$ 
	as $\epsilon\to+0$.
	Since 
	\begin{align*}
		\eta(e^\epsilon-e^{\eta\epsilon(1+\sqrt{\epsilon})})
		&= \eta(\epsilon-\eta\epsilon(1+\sqrt{\epsilon})) + O(\eta\epsilon^2)\\
		&= \eta(1-\eta)\epsilon - \eta^2\epsilon^{3/2} + O(\eta\epsilon^2)\quad\text{and}\\
		(1-\eta)(e^{\eta\epsilon(1+\sqrt{\epsilon})}-1)
		&= (1-\eta)\eta\epsilon(1+\sqrt{\epsilon}) + O(\eta\epsilon^2)\\
		&= \eta(1-\eta)\epsilon + \eta(1-\eta)\epsilon^{3/2} + O(\eta\epsilon^2)
	\end{align*}
	as $\epsilon\to+0$, we obtain 
	$\eta(e^\epsilon-e^{\eta\epsilon(1+\sqrt{\epsilon})}) < (1-\eta)(e^{\eta\epsilon(1+\sqrt{\epsilon})}-1)$ 
	as $\epsilon\to+0$.
\end{proof}

We now prove Theorem~\ref{main1'} in two steps: 
first for the asymmetric case, and then for the symmetric case.
Theorems~\ref{main2'} and \ref{main3'} are proved in the course of the proof of Theorem~\ref{main1'}.

\begin{proof}[Proof of Theorem~$\ref{main1'}$ for $A_\C^\eta(n,\epsilon)$ and $A_\Q^\eta(n,\epsilon)$]
	Let $\vec{\rho}$ be an $\epsilon$-QLDP mechanism.
	Recall that $\rho_{\avg}$ is defined as $\rho_{\avg}=(1/n)\sum_{x\in[n]} \rho_x$.
	Set $\delta\rho_x=\rho_x-\rho_{\avg}$ for $x\in[n]$.
	Define the utility function $\Psi_\Q$ as 
	\begin{equation*}
		\Psi_\Q(\vec{\rho})
		= \frac{1}{n}\sum_{x\in[n]} D( \rho_x \| \rho_{\avg} )
		= \chi(p_{\mix}; \vec{\rho}),
	\end{equation*}
	where $\chi(p_{\mix}; \vec{\rho})$ is the Holevo information \eqref{eq09}.
	\par
	For $k\in[n]$, define $\tilde{\rho}_k$ as \eqref{eq11}, i.e., 
	$\tilde{\rho}_k=\eta\rho_k+(1-\eta)\rho_{\avg}$.
	Lemma~\ref{lemQLDP} implies that the collection of $\tilde{\rho}_1,\ldots,\tilde{\rho}_n$ is $\eta\epsilon(1+\sqrt{\epsilon})$-QLDP as $\epsilon\to+0$.
	By $(1/n)\sum_{x\in[n]} \tilde{\rho}_x=\rho_{\avg}$, we have 
	\[
	A^\eta(\vec{\rho})
	= \min_{x\in[n]} D( \tilde{\rho}_x \| \rho_{\avg} )
	\le \Psi_\Q(\tilde{\rho}_1,\ldots,\tilde{\rho}_n).
	\]
	Therefore, Corollary~\ref{coro1} implies that 
	\begin{equation*}
		A_\Q^\eta(n,\epsilon)
		\le \OPT_n(\eta\epsilon(1+\sqrt{\epsilon}); \Psi_\Q)
		= \frac{n-1}{4n}\eta^2\epsilon^2 + o(\eta^2\epsilon^2)
	\end{equation*}
	as $\epsilon\to+0$.
	Also, for the isoclinic mechanism $\vec{\sigma}^*$, 
	it can easily be checked that 
	(i) $(1/n)\sum_{x\in[n]} \sigma_x^*=d^{-1}I_d$, and 
	(ii) $D( \tilde{\sigma}_x^* \| d^{-1}I_d )$ is independent of $x\in[n]$.
	These and Corollary~\ref{coro2} imply that 
	\[
	A_\Q^\eta(n,\epsilon)
	\ge A^\eta(\vec{\sigma}^*)
	= \Psi_\Q(\tilde{\sigma}_1^*,\ldots,\tilde{\sigma}_n^*)
	= \frac{n-1}{4n}\eta^2\epsilon^2 + o(\eta^2\epsilon^2)
	\]
	as $\epsilon\to+0$.
	Therefore, \eqref{eqAQ} holds.
	We also obtain \eqref{eqAC} in the same way.
\end{proof}

\begin{proof}[Proof of Theorem~$\ref{main1'}$ for $S_\C^\eta(n,\epsilon)$ and $S_\Q^\eta(n,\epsilon)$]
	Let $\vec{\rho}$ be an $\epsilon$-QLDP mechanism.
	Recall that $\rho_{\avg}$ is defined as $\rho_{\avg}=(1/n)\sum_{x\in[n]} \rho_x$.
	Set $\delta\rho_x=\rho_x-\rho_{\avg}$ for $x\in[n]$.
	Define the utility functions $\Psi_\Q,\Psi_\C$ and the sublinear function $\psi$ as 
	\begin{align*}
		\Psi_\Q(\vec{\rho}) &= \frac{-1}{n(n-1)}\sum_{i\not=j} \Tr\rho_i^{1/2}\rho_j^{1/2},\\
		\Psi_\C(q) &= \sum_{y:\,q(y|1),\ldots,q(y|n)>0} \psi(q(y|1),\ldots,q(y|n)),\quad\text{and}\\
		\psi(\mathbf{z}) &= \frac{-1}{n(n-1)}\sum_{i\not=j} \sqrt{z_iz_j},
	\end{align*}
	where $z_{\avg}=(1/n)\sum_{i\in[n]} z_i$.
	Then $\Psi_\Q$ is a quantum extension of $\Psi_\C$ (but it does not necessarily satisfy the data processing inequality).
	By \eqref{eq01}, it follows that $\|\delta\rho_x\|=O(\epsilon)$ for every $x\in[n]$.
	By Lemma~\ref{lemTr}, the utility function $\Psi_\Q$ has the second-order Taylor expansion around $(\rho_{\avg},\ldots,\rho_{\avg})$: 
	\begin{align*}
		\Psi_\Q(\vec{\rho})
		&= \frac{-1}{n(n-1)}\sum_{i\not=j} \biggl( 1-\frac{1}{8}J_{\rho_0}^{\WYD(1/2)}[\delta\rho_i-\delta\rho_j,\delta\rho_i-\delta\rho_j]
		+ o\Bigl( \max_{i\in[n]} \|\eta\delta\rho_i\|^2 \Bigr) \biggr)\\
		&= -1 + \frac{1}{8n(n-1)}\sum_{i\not=j} J_{\rho_0}^{\WYD(1/2)}[\delta\rho_i-\delta\rho_j,\delta\rho_i-\delta\rho_j]
		+ o(\epsilon^2)\\
		&= -1 + \frac{1}{4(n-1)}\sum_{i\in[n]} J_{\rho_0}^{\WYD(1/2)}[\delta\rho_i,\delta\rho_i]
		+ o(\epsilon^2),
	\end{align*}
	where we have used $\sum_{x\in[n]} \delta\rho_x=0$ to obtain the last equality.
	By this and $\partial_1^2\psi(\mathbbm{1}_n)=(2n)^{-1}>0$, 
	the functions $\Psi_\Q$ and $\psi$ satisfy all the assumptions of Theorem~\ref{main1}.
	\par
	For $k\in[n]$, define $\tilde{\rho}_k$ as \eqref{eq11}, i.e., 
	$\tilde{\rho}_k=\eta\rho_k+(1-\eta)\rho_{\avg}$.
	Lemma~\ref{lemQLDP} implies that the collection of $\tilde{\rho}_1,\ldots,\tilde{\rho}_n$ is $\eta\epsilon(1+\sqrt{\epsilon})$-QLDP as $\epsilon\to+0$.
	By Lemmas~\ref{lemChernoff} and \ref{lemTr}, 
	\begin{align*}
		S^\eta(\vec{\rho})
		&= \min_{i\not=j} C( \tilde{\rho}_i, \tilde{\rho}_j )
		= \min_{i\not=j} \Bigl( 1-\Tr\tilde{\rho}_i^{1/2}\tilde{\rho}_j^{1/2} \Bigr)
		+ o\Bigl( \max_{i\in[n]} \|\eta\delta\rho_i\|^2 \Bigr)\\
		&\le 1+\Psi_\Q(\tilde{\rho}_1,\ldots,\tilde{\rho}_n) + o(\eta^2\epsilon^2)
	\end{align*}
	as $\epsilon\to+0$.
	Therefore, Theorem~\ref{main1} implies that 
	\begin{align*}
		S_\Q^\eta(n,\epsilon)
		&\le 1+\OPT_n(\eta\epsilon(1+\sqrt{\epsilon}); \Psi_\Q) + o(\eta^2\epsilon^2)
		= \frac{n}{4}\cdot\frac{1}{2n}\eta^2\epsilon^2 + o(\eta^2\epsilon^2)\\
		&= \frac{1}{8}\eta^2\epsilon^2 + o(\eta^2\epsilon^2)
	\end{align*}
	as $\epsilon\to+0$.
	Also, for the isoclinic mechanism $\vec{\sigma}^*$, 
	it can easily be checked that 
	$\Tr(\tilde{\sigma}_i^*)^{1/2}(\tilde{\sigma}_j^*)^{1/2}$ is independent of $i\not=j$.
	This, Theorem~\ref{main2}, Lemmas~\ref{lemChernoff} and \ref{lemTr} imply that 
	\begin{align*}
		S_\Q^\eta(n,\epsilon)
		&\ge S^\eta(\vec{\sigma}^*)
		= \min_{i\not=j} \Bigl( 1-\Tr(\tilde{\sigma}_i^*)^{1/2}(\tilde{\sigma}_j^*)^{1/2} \Bigr) + o(\eta^2\epsilon^2)\\
		&= 1+\Psi_\Q(\tilde{\sigma}_1^*,\ldots,\tilde{\sigma}_n^*) + o(\eta^2\epsilon^2)
		= \frac{1}{8}\eta^2\epsilon^2 + o(\eta^2\epsilon^2)
	\end{align*}
	as $\epsilon\to+0$.
	Therefore, \eqref{eqSQ} holds.
	We also obtain \eqref{eqSC} in the same way.
\end{proof}

\begin{proof}[Proofs of Theorems~$\ref{main2'}$ and $\ref{main3'}$]
	The desired equalities have already been proved in the proof of Theorem~\ref{main1'}.
\end{proof}

\section{Proofs of Lemmas~$\ref{lemC1}$ and $\ref{lemQ1}$}\label{appD}
This appendix proves Lemmas~\ref{lemC1} and \ref{lemQ1}.
As before, let $n,d\ge2$, $r\in[d-1]$, and $\epsilon>0$, 
representing cardinality, dimension, rank, and privacy parameter, respectively.
As stated in Section~\ref{exact}, regarding Lemma~\ref{lemC1}, 
we only prove the inequality $\underline{S}_\C^\eta(n,\epsilon) \le S_\C^\eta(n,\epsilon)$.

\begin{proof}[Proof of $\underline{S}_\C^\eta(n,\epsilon) \le S_\C^\eta(n,\epsilon)$]
	Fix an integer $k\in[n-1]$, and let $\cS(n,k)$ be the set of all sets consisting of $k$ distinct elements in $[n]$.
	Define the $\epsilon$-LDP mechanism $Q=[Q(A|x)]_{A\in\cS(n,k),x\in[n]}$ as 
	\[
	Q(A|x) = \frac{1+(e^\epsilon-1)\mathbbm{1}_A(x)}{\binom{n-1}{k-1}e^\epsilon + \binom{n-1}{k}},
	\]
	where $\mathbbm{1}_A$ denotes the indicator function of a set $A$.
	Recall that $\tilde{q}$ is defined as $\tilde{q}=q[p_1^\eta,\ldots,p_n^\eta]$.
	From this, $\binom{n-1}{k}=\frac{n-k}{k}\binom{n-1}{k-1}$, and $\frac{n}{k}\binom{n-1}{k-1}=\binom{n}{k}$, 
	it follows that 
	\begin{equation}
	\begin{split}
		\tilde{Q}(A|x')
		&= \sum_{x\in[n]} Q(A|x)p_{x'}^\eta(x)
		= \eta Q(A|x') + \frac{1-\eta}{n}\sum_{x\in[n]} Q(A|x)\\
		&= \eta Q(A|x') + \frac{1-\eta}{n}\cdot\frac{n+(e^\epsilon-1)k}{\binom{n-1}{k-1}e^\epsilon + \binom{n-1}{k}}\\
		&= \eta Q(A|x') + (1-\eta)\binom{n}{k}^{-1}
		= \Delta_{k/n}^\eta(\mathbbm{1}_A(x'))\binom{n}{k}^{-1}
	\end{split}\label{eq37}
	\end{equation}
	for every $x'\in[n]$ and every $A\in\cS(n,k)$.
	Thus, 
	\begin{align*}
		\Tr\Diag_i(\tilde{Q})^s\Diag_j(\tilde{Q})^{1-s}
		&= \sum_{A\in\cS(n,k)} \tilde{Q}(A|i)^s\tilde{Q}(A|j)^{1-s}\\
		&= \binom{n}{k}^{-1}\biggl( \sum_{\substack{A\in\cS(n,k) \\ A\ni i,j}} \Delta_{k/n}^\eta(1)
		+ \sum_{\substack{A\in\cS(n,k) \\ A\not\ni i,j}} \Delta_{k/n}^\eta(0)\\
		&\qquad+ \sum_{\substack{A\in\cS(n,k) \\ A\ni i,A\not\ni j}} \Delta_{k/n}^\eta(1)^s\Delta_{k/n}^\eta(0)^{1-s}
		+ \sum_{\substack{A\in\cS(n,k) \\ A\not\ni i,A\ni j}} \Delta_{k/n}^\eta(0)^s\Delta_{k/n}^\eta(1)^{1-s} \biggr)
	\end{align*}
	for all $i\not=j$ and all $s\in[0,1]$.
	By the AM-GM inequality, the above function of $s\in[0,1]$ achieves the minimum at $s=1/2$. 
	This and \eqref{eq37} yield that 
	\begin{align*}
		\min_{s\in[0,1]} \Tr\Diag_i(\tilde{Q})^s\Diag_j(\tilde{Q})^{1-s}
		&= \sum_{A\in\cS(n,k)} \tilde{Q}(A|i)^{1/2}\tilde{Q}(A|j)^{1/2}\\
		&= 1 - \frac{1}{2}\sum_{A\in\cS(n,k)} \Bigl( \sqrt{\tilde{Q}(A|i)} - \sqrt{\tilde{Q}(A|j)} \Bigr)^2\\
		&= 1 - \frac{1}{2}\binom{n}{k}^{-1}\sum_{\substack{A\in\cS(n,k) \\ \mathbbm{1}_A(i)+\mathbbm{1}_A(j)=1}} \Bigl( \sqrt{\Delta_{k/n}^\eta(1)} - \sqrt{\Delta_{k/n}^\eta(0)} \Bigr)^2\\
		&= 1 - \frac{1}{2}\binom{n}{k}^{-1}\cdot2\binom{n-2}{k-1}\Bigl( \sqrt{\Delta_{k/n}^\eta(1)} - \sqrt{\Delta_{k/n}^\eta(0)} \Bigr)^2\\
		&= 1 - \frac{k(n-k)}{n(n-1)}\Bigl( \sqrt{\Delta_{k/n}^\eta(1)} - \sqrt{\Delta_{k/n}^\eta(0)} \Bigr)^2
	\end{align*}
	for all $i\not=j$.
	By this and \eqref{eq10'}, we obtain 
	\begin{align*}
		S_\C^\eta(n,\epsilon)
		&\ge S^\eta(Q) = S(Q;p_1^\eta,\ldots,p_n^\eta)\\
		&= -\ln\biggl( 1 - \frac{k(n-k)}{n(n-1)}\Bigl( \sqrt{\Delta_{k/n}^\eta(1)} - \sqrt{\Delta_{k/n}^\eta(0)} \Bigr)^2 \biggr).
	\end{align*}
	Since $k\in[n-1]$ is arbitrary, the desired inequality holds.
\end{proof}

We next prove Lemma~\ref{lemQ1} in two steps: 
first for the symmetric case, and then for the asymmetric case.

\begin{proof}[Proof of Lemma~$\ref{lemQ1}$ for $S^\eta(\vec{\sigma})$]
	For simplicity, write $\tilde{\sigma}_k$ and $\mu^{[\eta]}$ as $\sigma_k$ and $\mu$, respectively, in the following calculation.
	Let us calculate 
	\[
	\max_{i\not=j} \min_{s\in[0,1]} \Tr(\sigma_i^s\sigma_j^{1-s}).
	\]
	By \eqref{eq16}, 
	\begin{align*}
		\sigma_i^s &= \Bigl( \frac{\mu}{d} \Bigr)^sI_d + \Bigl( \Bigl( \frac{\mu}{d}+\frac{1-\mu}{r} \Bigr)^s - \Bigl( \frac{\mu}{d} \Bigr)^s \Bigr)P_i
		\quad\text{and}\\
		\sigma_j^{1-s} &=  \Bigl( \frac{\mu}{d} \Bigr)^{1-s}I_d + \Bigl( \Bigl( \frac{\mu}{d}+\frac{1-\mu}{r} \Bigr)^{1-s} - \Bigl( \frac{\mu}{d} \Bigr)^{1-s} \Bigr)P_j.
	\end{align*}
	This yields 
	\begin{align*}
		\Tr(\sigma_i^s\sigma_j^{1-s})
		&= \frac{\mu}{d}\Tr I_d
		+ \Bigl( \frac{\mu}{d} \Bigr)^{1-s}\Bigl( \Bigl( \frac{\mu}{d}+\frac{1-\mu}{r} \Bigr)^s - \Bigl( \frac{\mu}{d} \Bigr)^s \Bigr)\Tr P_i\\
		&\quad+ \Bigl( \frac{\mu}{d} \Bigr)^s\Bigl( \Bigl( \frac{\mu}{d}+\frac{1-\mu}{r} \Bigr)^{1-s} - \Bigl( \frac{\mu}{d} \Bigr)^{1-s} \Bigr)\Tr P_j\\
		&\quad+ \Bigl( \Bigl( \frac{\mu}{d}+\frac{1-\mu}{r} \Bigr)^s - \Bigl( \frac{\mu}{d} \Bigr)^s \Bigr)
		\Bigl( \Bigl( \frac{\mu}{d}+\frac{1-\mu}{r} \Bigr)^{1-s} - \Bigl( \frac{\mu}{d} \Bigr)^{1-s} \Bigr)\Tr P_iP_j.
	\end{align*}
	By $\Tr P_i=\Tr P_j=r$ and $\Tr P_iP_j=rc$, we have 
	\begin{align*}
		\Tr(\sigma_i^s\sigma_j^{1-s})
		&= (d-2r)\frac{\mu}{d}
		+ r\Bigl( \frac{\mu}{d} \Bigr)^{1-s}\Bigl( \frac{\mu}{d}+\frac{1-\mu}{r} \Bigr)^s
		+ r\Bigl( \frac{\mu}{d} \Bigr)^s\Bigl( \frac{\mu}{d}+\frac{1-\mu}{r} \Bigr)^{1-s}\\
		&\quad+ rc\Bigl( \frac{2\mu}{d}+\frac{1-\mu}{r} \Bigr)
		- rc\Bigl( \frac{\mu}{d} \Bigr)^{1-s}\Bigl( \frac{\mu}{d}+\frac{1-\mu}{r} \Bigr)^s
		- rc\Bigl( \frac{\mu}{d} \Bigr)^s\Bigl( \frac{\mu}{d}+\frac{1-\mu}{r} \Bigr)^{1-s}\\
		&= (1-2u)\mu + c(2u\mu+1-\mu)\\
		&\quad+ r(1-c)\Bigl( \frac{\mu}{d} \Bigr)^{1-s}\Bigl( \frac{\mu}{d}+\frac{1-\mu}{r} \Bigr)^s
		+ r(1-c)\Bigl( \frac{\mu}{d} \Bigr)^s\Bigl( \frac{\mu}{d}+\frac{1-\mu}{r} \Bigr)^{1-s},
	\end{align*}
	where $u=r/d$.
	By the AM-GM inequality, the right-hand side achieves the minimum at $s=1/2$: 
	\begin{equation*}
	\begin{split}
		\min_{s\in[0,1]} \Tr(\sigma_i^s\sigma_j^{1-s})
		&= (1-2u)\mu + c(2u\mu+1-\mu)
		+ 2r(1-c)\Bigl( \frac{\mu}{d} \Bigr)^{1/2}\Bigl( \frac{\mu}{d}+\frac{1-\mu}{r} \Bigr)^{1/2}\\
		&= (1-2u)\mu + c(2u\mu+1-\mu) + 2(1-c)\sqrt{u\mu(u\mu+1-\mu)}\\
		&= 1-G_\sym(\mu).
	\end{split}
	\end{equation*}
	Since the right-hand side is independent of $i\not=j$, 
	it turns out that 
	\[
	\max_{i\not=j} \min_{s\in[0,1]} \Tr(\sigma_i^s\sigma_j^{1-s})
	= 1-G_\sym(\mu).
	\]
	This and \eqref{eq10} imply the desired equality.
\end{proof}

\begin{proof}[Proof of Lemma~$\ref{lemQ1}$ for $A^\eta(\vec{\sigma})$]
	Let us calculate $A^\eta(\vec{\sigma})$.
	By \eqref{eq12} and $\sigma_{\avg}=d^{-1}I_d$, 
	\[
	A^\eta(\vec{\sigma}) = \min_{x\in[n]} D(\tilde{\sigma}_x\|d^{-1}I_d).
	\]
	Since $\tilde{\sigma}_x$ and $d^{-1}I_d$ are commutative, 
	the unitary invariance of the Umegaki relative entropy and $\tilde{\sigma}_x=\mu^{[\eta]}d^{-1}I_d+(1-\mu^{[\eta]})r^{-1}P_x$ imply 
	\[
	A^\eta(\vec{\sigma}) = D\Bigl( \frac{\mu^{[\eta]}}{d}\mathbbm{1}_d+\frac{1-\mu^{[\eta]}}{r}(\mathbf{e}_1+\cdots+\mathbf{e}_r) \Big\| \frac{1}{d}\mathbbm{1}_d \Bigr),
	\]
	where $(\mathbf{e}_i)_{i\in[d]}$ is the standard basis of $\mathbb{R}^d$.
	For simplicity, write $\mu^{[\eta]}$ as $\mu$ in the following calculation.
	Then 
	\begin{align*}
		A^\eta(\vec{\sigma})
		&= r\Bigl( \frac{\mu}{d}+\frac{1-\mu}{r} \Bigr)\Bigl( \ln\Bigl( \frac{\mu}{d}+\frac{1-\mu}{r} \Bigr)-\ln\frac{1}{d} \Bigr)
		+ (d-r)\frac{\mu}{d}\Bigl( \ln\frac{\mu}{d}-\ln\frac{1}{d} \Bigr)\\
		&= \ln d + rL\Bigl( \frac{\mu}{d}+\frac{1-\mu}{r} \Bigr) + (d-r)L\Bigl( \frac{\mu}{d} \Bigr).
	\end{align*}
	This implies 
	\begin{align*}
		A^\eta(\vec{\sigma}) &= \ln d + r\Bigl( \frac{\mu}{d}+\frac{1-\mu}{r} \Bigr)\log\Bigl( \frac{\mu}{d}+\frac{1-\mu}{r} \Bigr)
		+ (d-r)\Bigl( \frac{\mu}{d} \Bigr)\ln\Bigl( \frac{\mu}{d} \Bigr)\\
		&= (u\mu+1-\mu)\log(\mu+u^{-1}(1-\mu)) + (1-u)\mu\ln\mu\\
		&= uL(\mu+u^{-1}(1-\mu)) + (1-u)L(\mu)
		= G_\asym(\mu),
	\end{align*}
	where $u=r/d$.
\end{proof}

\section{Proof of Proposition~$\ref{prop2}$}\label{appE}
This appendix proves Proposition~\ref{prop2}, more precisely, \eqref{eq56} and \eqref{eq58}.
First, we show \eqref{eq58}.
Let $c=c(u)$, $\mu=\mu(u)$, and $\mu_0=\mu_0(u)$ be the functions of $u$ in \eqref{eq17} and \eqref{eq38}.
Let $\vec{\sigma}$ be an isoclinic mechanism in \eqref{eq13}, where $c$ and $\mu$ are evaluated at $u=r/d$.
For $x\in[n]$, define the quantum state $\sigma_x^{(0)}$ as 
\[
\sigma_x^{(0)} = \frac{\mu_0}{d}I_d + \frac{1-\mu_0}{r}P_x,
\]
where $\mu_0$ is evaluated at $u=r/d$.
By $\mu\le\mu_0<1$, the depolarizing channel $\Lambda(X)=t(\Tr X)d^{-1}I_d+(1-t)X$ with $1-t=(1-\mu_0)(1-\mu)^{-1}$ satisfies that 
$\Lambda(\sigma_x)=\sigma_x^{(0)}$ for every $x\in[n]$.
This and the data processing inequality yield 
\begin{equation}
\begin{split}
	S^\eta(\vec{\sigma}^{(0)}) &\le S^\eta(\vec{\sigma}) \le S_\iso^\eta(n,\epsilon)
	\quad\text{and}\\
	A^\eta(\vec{\sigma}^{(0)}) &\le A^\eta(\vec{\sigma}) \le A_\iso^\eta(n,\epsilon).
\end{split}\label{eq41}
\end{equation}
By Lemma~\ref{lemQ1} and \eqref{eq-G-C}, 
\begin{equation}
\begin{split}
	S^\eta(\vec{\sigma}^{(0)})
	&= -\ln\bigl( 1-G_\sym(\mu_0^{[\eta]}(r/d)) \bigr)
	= -\ln\bigl( 1-G_\sym^{\C;\eta}(r/d) \bigr)
	\quad\text{and}\\
	A^\eta(\vec{\sigma}^{(0)})
	&= G_\asym(\mu_0^{[\eta]}(r/d))
	= G_\asym^{\C;\eta}(r/d).
\end{split}\label{eq42}
\end{equation}
By \eqref{eq41} and \eqref{eq42}, we obtain \eqref{eq58}.
\par
Next, we show the first equality in \eqref{eq56}.
By \eqref{eq38} and \eqref{eq43}, 
\begin{equation}
\begin{split}
	\mu_0^{[\eta]}
	&= \eta\mu_0+1-\eta
	= \frac{\eta}{u(e^\epsilon-1)+1}+1-\eta
	= \Delta_u^\eta(0)
	\quad\text{and}\\
	\mu_0^{[\eta]}+u^{-1}(1-\mu_0^{[\eta]})
	&= \mu_0^{[\eta]}+u^{-1}\eta(1-\mu_0)
	= \mu_0^{[\eta]} + \frac{\eta(e^\epsilon-1)}{u(e^\epsilon-1)+1}
	= \Delta_u^\eta(1),
\end{split}\label{eq39}
\end{equation}
where $\Delta_u^\eta(t)$ and $\mu_0^{[\eta]}$ are defined in Lemmas~\ref{lemC1} and \ref{lemQ1}, respectively.
Lemma~\ref{lemQ1}, \eqref{eq-G-C}, \eqref{eq17}, and \eqref{eq39} yield 
\begin{equation}
\begin{split}
	G_\sym^{\C;\eta}(u)
	&= G_\sym(\mu_0^{[\eta]})
	= (1-c)\Bigl( \sqrt{u\mu_0^{[\eta]}+1-\mu_0^{[\eta]}}-\sqrt{u\mu_0^{[\eta]}} \Bigr)^2\\
	&= \frac{nu(1-u)}{n-1}\Bigl( \sqrt{\Delta_u^\eta(1)}-\sqrt{\Delta_u^\eta(0)} \Bigr)^2.
\end{split}\label{eq44}
\end{equation}
By \eqref{eq38}, we have 
$\Delta_u^\eta(1)=\eta e^\epsilon\mu_0+1-\eta$ and $\Delta_u^\eta(0)=\eta\mu_0+1-\eta$.
Since $\sqrt{\Delta_u^\eta(1)}-\sqrt{\Delta_u^\eta(0)}$ can be written as 
\[
\sqrt{\Delta_u^\eta(1)}-\sqrt{\Delta_u^\eta(0)}
= \frac{\eta(e^\epsilon-1)\sqrt{\mu_0}}{\sqrt{\eta e^\epsilon+(1-\eta)\mu_0^{-1}}+\sqrt{\eta+(1-\eta)\mu_0^{-1}}},
\]
it increases in $\mu_0$, i.e, decreases in $u$.
Thus, 
\begin{equation}
	G_\sym^{\C;\eta}(u) \ge G_\sym^{\C;\eta}(1-u)
	\label{eq45}
\end{equation}
for every $u\in[0,1/2]$.
Lemma~\ref{lemC1}, \eqref{eq44}, and \eqref{eq45} imply 
\begin{equation*}
	\underline{S}_\C^\eta(n,\epsilon)
	= \max_{k\in[n-1]} (-1)\ln\bigl( 1-G_\sym^{\C;\eta}(k/n) \bigr)
	= \max_{k\in[1,n/2]\cap\mathbb{Z}} (-1)\ln\bigl( 1-G_\sym^{\C;\eta}(k/n) \bigr).
\end{equation*}
Therefore, the first equality in \eqref{eq56} holds.
\par
Finally, we show the second equality in \eqref{eq56}.
It suffices to show the counterparts of \eqref{eq44} and \eqref{eq45}.
Lemma~\ref{lemQ1}, \eqref{eq-G-C}, and \eqref{eq39} yield 
\begin{equation}
\begin{split}
	G_\asym^{\C;\eta}(u)
	&= G_\asym(\mu_0^{[\eta]})
	= uL(\mu_0^{[\eta]}+u^{-1}(1-\mu_0^{[\eta]})) + (1-u)L(\mu_0^{[\eta]})\\
	&= uL(\Delta_u^\eta(1)) + (1-u)L(\Delta_u^\eta(0)).
\end{split}\label{eq57}
\end{equation}
By Lemma~\ref{lemC1}, this is the counterpart of \eqref{eq44}.
The remainder is the counterpart of \eqref{eq45}, namely, 
\begin{equation}
	G_\asym^{\C;\eta}(u) \ge G_\asym^{\C;\eta}(1-u)
	\label{eq46}
\end{equation}
for every $u\in[0,1/2]$.
While the following formulation is not strictly necessary to prove \eqref{eq46}, 
it provides useful intuitive insight.
Let $\mathbb{P}_X=[\mathbb{P}_X(0),\mathbb{P}_X(1)]^\top=[1-u,u]^\top$ and 
\[
\mathbb{P}_{Y|X} = [\mathbb{P}_{Y|X}(y|x)]_{y,x\in\{0,1\}} = \frac{1}{e^\epsilon+1}
\begin{bmatrix}
	e^\epsilon&1\\
	1&e^\epsilon
\end{bmatrix}.
\]
Set $Z(k)=\mathbb{P}_X(k)(e^\epsilon-1)+1$ for $k\in\{0,1\}$. Then 
\begin{gather*}
	\mathbb{P}_{X|Y=0} = \frac{1}{Z(0)}
	\begin{bmatrix}
		(1-u)e^\epsilon\\
		u
	\end{bmatrix}
	,\quad
	\mathbb{P}_{X|Y=1} = \frac{1}{Z(1)}
	\begin{bmatrix}
		1-u\\
		ue^\epsilon
	\end{bmatrix}
	,\\
	\mathbb{P}_0^\eta \coloneqq \eta\mathbb{P}_{X|Y=0}+(1-\eta)\mathbb{P}_X = 
	\begin{bmatrix}
		(1-u)\Delta_{1-u}^\eta(1)\\
		u\Delta_{1-u}^\eta(0)
	\end{bmatrix}
	,\quad\text{and}\\
	\mathbb{P}_1^\eta \coloneqq \eta\mathbb{P}_{X|Y=1}+(1-\eta)\mathbb{P}_X = 
	\begin{bmatrix}
		(1-u)\Delta_u^\eta(0)\\
		u\Delta_u^\eta(1)
	\end{bmatrix}.
\end{gather*}
Moreover, 
\begin{equation*}
	G_\asym^{\C;\eta}(u) = D(\mathbb{P}_1^\eta \| \mathbb{P}_X)
	\quad\text{and}\quad
	G_\asym^{\C;\eta}(1-u) = D(\mathbb{P}_0^\eta \| \mathbb{P}_X).
\end{equation*}
This and Lemma~\ref{lem4} (which is proved below) imply \eqref{eq46} for every $u\in[0,1/2]$.\qed

\begin{lemma}\label{lem4}
	Let $F\colon (0,\infty)\to\mathbb{R}$ be an operator convex function satisfying $F(1)=0$, 
	and let $u\in[0,1/2]$.
	Then $D_F(\mathbb{P}_1^\eta \| \mathbb{P}_X) \ge D_F(\mathbb{P}_0^\eta \| \mathbb{P}_X)$ holds 
	for the above probability distributions $\mathbb{P}_X$, $\mathbb{P}_0^\eta$, and $\mathbb{P}_1^\eta$.
	In particular, $D(\mathbb{P}_1^\eta \| \mathbb{P}_X) \ge D(\mathbb{P}_0^\eta \| \mathbb{P}_X)$ holds.
\end{lemma}

We can intuitively interpret Lemma~\ref{lem4} as follows.
Since $X=1$ is a rarer event than $X=0$, 
observing $Y=1$ (which suggests $X=1$ more strongly) yields a greater update in information about $X$ than observing $Y=0$.

\begin{proof}[Proof of Lemma~\upshape{\ref{lem4}}]
	By \cite[Eq.~(1.4)]{Nguyen} (cf.\ \cite[Problem~V.5.5]{book4}), the operator convex function $F$ is generated 
	by a linear term and the functions $(t-1)^2$ and $(t-1)^2(t+s)^{-1}$, $s\ge0$.
	Since $(t-1)^2(t+s)^{-1}=s^{-1}+t-s^{-1}(s+1)^2t(t+s)^{-1}$, 
	the operator convex function $F$ is generated 
	by a linear term and the functions $t^2$ and $-t(t+s)^{-1}$, $s\ge0$.
	Thus, it suffices to show $D_F(\mathbb{P}_1^\eta \| \mathbb{P}_X) \ge D_F(\mathbb{P}_0^\eta \| \mathbb{P}_X)$ 
	in both cases $F(t)=t^2$ and $F(t)=-t(t+s)^{-1}$, where $s\ge0$.
	Note that $D_F$ can be defined as \eqref{eq04} even for operator convex functions that are not necessarily normalized.
	\par
	\textbf{Case $F(t)=t^2$.}
	We have 
	\begin{align*}
		D_F(\mathbb{P}_1^\eta \| \mathbb{P}_X)
		&= (1-u)F(\Delta_u^\eta(0)) + uF(\Delta_u^\eta(1))\\
		&= (1-u)\Bigl( \frac{\eta}{Z(1)} + 1-\eta \Bigr)^2 + u\Bigl( \frac{\eta e^\epsilon}{Z(1)} + 1-\eta \Bigr)^2\\
		&= \frac{\eta^2(1-u+ue^{2\epsilon})}{Z(1)^2} + \frac{2\eta(1-\eta)(1-u+ue^\epsilon)}{Z(1)} + (1-\eta)^2\\
		&= \frac{\eta^2(Z(1)+e^\epsilon(Z(1)-1))}{Z(1)^2} + 2\eta(1-\eta) + (1-\eta)^2\\
		&= \eta^2\Bigl( \frac{e^\epsilon+1}{Z(1)} - \frac{e^\epsilon}{Z(1)^2} \Bigr) + 1-\eta^2.
	\end{align*}
	Similarly, 
	\[
	D_F(\mathbb{P}_0^\eta \| \mathbb{P}_X)
	= \eta^2\Bigl( \frac{e^\epsilon+1}{Z(0)} - \frac{e^\epsilon}{Z(0)^2} \Bigr) + 1-\eta^2.
	\]
	These and $Z(0)+Z(1)=e^\epsilon+1$ yield 
	\begin{align*}
		&\quad D_F(\mathbb{P}_1^\eta \| \mathbb{P}_X) - D_F(\mathbb{P}_0^\eta \| \mathbb{P}_X)\\
		&= \eta^2\Bigl( \frac{1}{Z(1)} - \frac{1}{Z(0)} \Bigr)\biggl( e^\epsilon+1 - e^\epsilon\Bigl( \frac{1}{Z(1)} + \frac{1}{Z(0)} \Bigr) \biggr)\\
		&= \eta^2\frac{Z(0)-Z(1)}{Z(0)Z(1)}(e^\epsilon+1)\Bigl( 1 - \frac{e^\epsilon}{Z(0)Z(1)} \Bigr).
	\end{align*}
	By this, $Z(0)\ge Z(1)$, and $Z(0)Z(1)\ge e^\epsilon$, we obtain the desired inequality.
	\par
	\textbf{Case $F(t)=-t(t+s)^{-1}$.}
	Set $\theta=1-\eta+s$.
	Since $(1-u)\Delta_u^\eta(0)+u\Delta_u^\eta(1)=1$ and 
	\[
	\frac{(1-u)t_0}{t_0+s} + \frac{ut_1}{t_1+s}
	= \frac{t_0t_1+s((1-u)t_0+ut_1)}{(t_0+s)(t_1+s)},
	\]
	we have 
	\begin{align*}
		D_F(\mathbb{P}_1^\eta \| \mathbb{P}_X)
		&= (1-u)F(\Delta_u^\eta(0)) + uF(\Delta_u^\eta(1))\\
		&= -\frac{\Delta_u^\eta(0)\Delta_u^\eta(1)+s}{(\Delta_u^\eta(0)+s)(\Delta_u^\eta(1)+s)}.
	\end{align*}
	The denominator is 
	\[
	(\eta Z(1)^{-1}+\theta)(\eta e^\epsilon Z(1)^{-1}+\theta),
	\]
	and the numerator is 
	\begin{align*}
		&\quad \bigl( \eta Z(1)^{-1}+1-\eta \bigr)\bigl( \eta e^\epsilon Z(1)^{-1}+1-\eta \bigr)+s\\
		&= \eta^2 e^\epsilon Z(1)^{-2}+\eta(1-\eta)(e^\epsilon+1)Z(1)^{-1}+\theta-\eta(1-\eta)\\
		&= \eta^2 e^\epsilon Z(1)^{-2}+\eta(1-\eta)Z(0)Z(1)^{-1}+\theta,
	\end{align*}
	where we have used $Z(0)+Z(1)=e^\epsilon+1$ to obtain the last equality.
	Setting $W$ as 
	\begin{equation}
		W = \eta^2 e^\epsilon+\eta(1-\eta)Z(0)Z(1),
		\label{eq53}
	\end{equation}
	we obtain 
	\[
	D_F(\mathbb{P}_1^\eta \| \mathbb{P}_X)
	= -\frac{W+\theta Z(1)^2}{(\eta+\theta Z(1))(\eta e^\epsilon+\theta Z(1))}.
	\]
	Similarly, 
	\[
	D_F(\mathbb{P}_0^\eta \| \mathbb{P}_X)
	= -\frac{W+\theta Z(0)^2}{(\eta+\theta Z(0))(\eta e^\epsilon+\theta Z(0))}.
	\]
	Thus, the sign of the difference $D_F(\mathbb{P}_1^\eta \| \mathbb{P}_X) - D_F(\mathbb{P}_0^\eta \| \mathbb{P}_X)$ coincides with 
	\begin{equation}
	\begin{split}
		&\quad-(W+\theta Z(1)^2)(\eta+\theta Z(0))(\eta e^\epsilon+\theta Z(0))\\
		&\quad+ (W+\theta Z(0)^2)(\eta+\theta Z(1))(\eta e^\epsilon+\theta Z(1))\\
		&= W\bigl( \theta^2(Z(1)^2-Z(0)^2) + \eta(e^\epsilon+1)\theta(Z(1)-Z(0)) \bigr)\\
		&\quad- \theta Z(1)^2\bigl( \eta^2 e^\epsilon + \eta(e^\epsilon+1)\theta Z(0) \bigr)
		+ \theta Z(0)^2\bigl( \eta^2 e^\epsilon + \eta(e^\epsilon+1)\theta Z(1) \bigr)\\
		&\overset{\text{(a)}}{=} W(e^\epsilon+1)(Z(1)-Z(0))(\theta^2 + \eta\theta)\\
		&\quad+ (e^\epsilon+1)(Z(0)-Z(1))( \theta\eta^2 e^\epsilon + \theta^2\eta Z(0)Z(1) )\\
		&= (e^\epsilon+1)(Z(1)-Z(0))\theta\eta\bigl( W\eta^{-1}(\theta + \eta) - (\eta e^\epsilon + \theta Z(0)Z(1)) \bigr),
	\end{split}\label{eq54}
	\end{equation}
	where we have used $Z(0)+Z(1)=e^\epsilon+1$ to obtain (a).
	Furthermore, it follows from \eqref{eq53} and $\theta=1-\eta+s$ that 
	\begin{align*}
		&\quad W\eta^{-1}(\theta + \eta) - (\eta e^\epsilon + \theta Z(0)Z(1))\\
		&= ( \eta e^\epsilon+(1-\eta)Z(0)Z(1) )(1+s) - (\eta e^\epsilon + (1-\eta+s)Z(0)Z(1))\\
		&= s\eta( e^\epsilon-Z(0)Z(1) ).
	\end{align*}
	By this, \eqref{eq54}, $Z(0)\ge Z(1)$, and $Z(0)Z(1)\ge e^\epsilon$, we obtain the desired inequality.
\end{proof}

\section{Proof of Theorem~$\ref{main4}$}\label{appF}
This appendix proves Theorem~\ref{main4}.
By Proposition~\ref{prop3}, $\underline{S}_\C^1(n,\epsilon)=S_\C^1(n,\epsilon)$, and \eqref{eq60}, 
it suffices to show the following implications: 
\begin{gather}
	\theta \le (n-1)^2
	\;\Longrightarrow\;
	G_\sym^{\C;1}(1/(n-1)) \ge G_\sym^{\C;1}(1/n)
	\quad\text{and} \label{eq61}\\
	\xi^2 \le (n-1)(n-2)
	\;\Longrightarrow\;
	G_\asym^{\C;1}(1/(n-1)) \ge G_\asym^{\C;1}(1/n), \label{eq62}
\end{gather}
where $\theta=e^\epsilon+n-1$ and $\xi=e^\epsilon$.
First, we show \eqref{eq61}.
Keeping Lemma~\ref{lemC2} and \eqref{eq56} in mind, 
we obtain 
\[
G_\sym^{\C;1}(u) = \frac{n(e^{\epsilon/2}-1)^2}{n-1}\cdot\frac{u(1-u)}{u(e^\epsilon-1)+1}
\]
by \eqref{eq44}. Thus, 
\begin{align*}
	G_\sym^{\C;1}(1/(n-1)) \ge G_\sym^{\C;1}(1/n)
	\iff \frac{1-1/(n-1)}{e^\epsilon+n-2} \ge \frac{1-1/n}{e^\epsilon+n-1}\\
	\iff \frac{n(n-2)}{\theta-1} \ge \frac{(n-1)^2}{\theta}
	\iff \theta \le (n-1)^2.
\end{align*}
This yields \eqref{eq61}.
\par
Next, we show \eqref{eq62}.
Keeping Lemma~\ref{lemC2} and \eqref{eq56} in mind, 
we obtain 
\[
G_\asym^{\C;1}(u) = \frac{uL(e^\epsilon) - L( u(e^\epsilon-1)+1 )}{u(e^\epsilon-1)+1}
= \frac{uL(e^\epsilon)}{u(e^\epsilon-1)+1} - \ln\bigl( u(e^\epsilon-1)+1 \bigr)
\]
by \eqref{eq57}. Thus, 
\begin{equation}
\begin{split}
	&\quad G_\asym^{\C;1}(1/(n-1)) \ge G_\asym^{\C;1}(1/n)\\
	&\iff \frac{L(\xi)}{\xi+n-2} - \ln\frac{\xi+n-2}{n-1} \ge \frac{L(\xi)}{\xi+n-1} - \ln\frac{\xi+n-1}{n}\\
	&\iff \frac{L(\xi)}{(\xi+n-2)(\xi+n-1)} \ge \int_1^\xi \frac{\mathrm{d}s}{(s+n-2)(s+n-1)}\\
	&\iff G(\xi) \ge 0,
\end{split}\label{eq63}
\end{equation}
where the functions $G$ and $g$ are defined as 
\[
G(t) = L(t)g(t) - \int_1^t g(s)\,\mathrm{d}s
\quad\text{and}\quad
g(t) = \frac{1}{(t+n-2)(t+n-1)},
\]
respectively. The derivative of $G$ is 
\begin{align*}
	G'(t) &= (\ln t+1)g(t) + (t\ln t)g'(t) - g(t)\\
	&= (\ln t)(tg(t))'.
\end{align*}
Moreover, 
\begin{equation*}
	(tg(t))' = \frac{(n-2)(n-1)-t^2}{(t+n-2)^2(t+n-1)^2}.
\end{equation*}
Therefore, if $\xi^2\le(n-1)(n-2)$, then 
\[
G(\xi) = G(\xi)-G(1)
= \int_1^\xi G'(s)\,\mathrm{d}s
\ge 0.
\]
This and \eqref{eq63} yield \eqref{eq62}.\qed

\section{Proof of Theorem~$\ref{main5}$}\label{appG}
This appendix proves Theorem~\ref{main5}.
As before, let $n,d\ge2$, $r\in[d-1]$, and $\epsilon>0$, 
representing cardinality, dimension, rank, and privacy parameter, respectively.
Since the second equality in Theorem~\ref{main5} follows from \cite[Theorem~2.2]{Yoshida}, 
it suffices to show the first equality in Theorem~\ref{main5}.
\par
Recall that $\vec{\sigma}^*$ is defined in \eqref{eq-opt-QLDP}, and 
that $\sigma_{i,j;\theta}^*=(1-\theta)\sigma_i^*+\theta\sigma_j^*$.
For all $i\not=j$, 
\begin{gather*}
	(\sigma_j^*-\sigma_i^*)^2
	= \Bigl( \frac{1-\mu_*}{r} \Bigr)^2(P_j-P_i)^2\quad\text{and}\\
	\sigma_{i,j;\theta}^* = \frac{\mu_*}{d}I_d + \frac{1-\mu_*}{r}\bigl( (1-\theta)P_i+\theta P_j \bigr).
\end{gather*}
By Lemmas~\ref{lemJordan}, \ref{lem2}, and $d=2r$, 
for all $i\not=j$, there exists a unitary matrix $U_{i,j}$ such that 
\[
U_{i,j}P_iU_{i,j}^\dag = \bigoplus_{k\in[r]}
\begin{bmatrix}
	1&0\\
	0&0
\end{bmatrix}
\quad\text{and}\quad
U_{i,j}P_jU_{i,j}^\dag = \bigoplus_{k\in[r]}
\begin{bmatrix}
	c&\sqrt{c(1-c)}\\
	\sqrt{c(1-c)}&1-c
\end{bmatrix},
\]
where $\bigoplus_{k\in[r]} X$ denotes the direct sum of $r$ copies of a matrix $X$.
Thus, for all $i\not=j$, 
\begin{gather}
	(P_j-P_i)^2 = (1-c)I_d,\quad
	U_{i,j}\sigma_{i,j;\theta}^*U_{i,j}^\dag
	= \bigoplus_{k\in[r]} R,\quad\text{and} \label{eq64}\\
	\Tr(\sigma_j^*-\sigma_i^*)^2(\sigma_{i,j;\theta}^*)^{-1}
	= \Bigl( \frac{1-\mu_*}{r} \Bigr)^2(1-c)\cdot r\Tr(R^{-1}) \label{eq65},
\end{gather}
where 
\[
R = \frac{\mu_*}{d}I_2 + \frac{1-\mu_*}{r}
\begin{bmatrix}
	1-\theta(1-c)&\theta\sqrt{c(1-c)}\\
	\theta\sqrt{c(1-c)}&\theta(1-c)
\end{bmatrix}.
\]
\par
Let us calculate $\Tr(R^{-1})$.
Since $R$ is of order two, it turns out that 
\[
\Tr(R^{-1}) = \frac{\Tr R}{\det R}.
\]
The trace and determinant of $R$ are calculated as follows: 
by \eqref{eq64}, 
\[
1 = \Tr\sigma_{i,j;\theta}^* = r\Tr R;
\]
setting $c'=1-c$, we have 
\begin{align*}
	\det R &= \Bigl( \frac{\mu_*}{d} + \frac{1-\mu_*}{r}(1-\theta c') \Bigr)
	\Bigl( \frac{\mu_*}{d} + \frac{1-\mu_*}{r}\theta c' \Bigr)
	- \Bigl( \frac{1-\mu_*}{r} \Bigr)^2\theta^2cc'\\
	&= \Bigl( \frac{\mu_*}{d} \Bigr)^2 + \Bigl( \frac{1-\mu_*}{r} \Bigr)^2(1-\theta c')\theta c' + \frac{\mu_*(1-\mu_*)}{dr}
	- \Bigl( \frac{1-\mu_*}{r} \Bigr)^2\theta^2cc'\\
	&= \Bigl( \frac{\mu_*}{d} \Bigr)^2 + \frac{\mu_*(1-\mu_*)}{dr}
	+ \Bigl( \frac{1-\mu_*}{r} \Bigr)^2(1-\theta)\theta c'\\
	&= \frac{\mu_*}{2r^2}\Bigl( 1-\frac{\mu_*}{2} \Bigr)
	+ \Bigl( \frac{1-\mu_*}{r} \Bigr)^2(1-\theta)\theta c',
\end{align*}
where the last equality follows from $d=2r$.
Therefore, 
\begin{equation}
	\Tr(R^{-1}) = \frac{r^2\Tr R}{r^2\det R}
	= \frac{r}{(\mu_*/2)(1-\mu_*/2) + (1-\mu_*)^2(1-\theta)\theta c'}.
	\label{eq66}
\end{equation}
\par
Now, set the positive number $\xi$ as 
\[
\xi = \frac{(\mu_*/2)(1-\mu_*/2)}{(1-\mu_*)^2(1-c)}.
\]
By \eqref{eq65} and \eqref{eq66}, 
\begin{align*}
	R_\Q^\theta(\vec{\sigma}^*)
	&= \min_{i\not=j} \Tr(\sigma_j^*-\sigma_i^*)^2(\sigma_{i,j;\theta}^*)^{-1}\\
	&= \frac{(1-\mu_*)^2(1-c)}{(\mu_*/2)(1-\mu_*/2)+(1-\mu_*)^2(1-\theta)\theta(1-c)}
	= \frac{1}{\xi+(1-\theta)\theta}.
\end{align*}
By \eqref{eq15}, 
\[
\xi = \frac{(1-\mu_*)^{-2}-1}{4(1-c)}
= \frac{1}{4\sinh^2(\epsilon/2)}
= \frac{e^\epsilon}{(e^\epsilon-1)^2}.
\]
Therefore, 
\begin{align*}
	R_\Q^\theta(\vec{\sigma}^*)
	&= \frac{1}{\xi+(1-\theta)\theta}
	= \frac{(e^\epsilon-1)^2}{e^\epsilon+(e^\epsilon-1)^2(1-\theta)\theta}\\
	&= \frac{(e^\epsilon-1)^2}{( (1-\theta)e^\epsilon+\theta )( \theta e^\epsilon+1-\theta )},
\end{align*}
which shows the first equality in Theorem~\ref{main5}.\qed

\section{Verification of the second-order Taylor expansion $\eqref{eqTaylorFdiv}$}\label{appTaylorFdiv}
This appendix verifies \eqref{eqTaylorFdiv}.
Let $\rho$ be a full-rank quantum state on $\mathbb{C}^d$, 
$X_1$ and $X_2$ be traceless Hermitian matrices on $\mathbb{C}^d$, and 
$F\colon (0,\infty)\to\mathbb{R}$ be an operator convex function satisfying $F(1)=0$.
For small real numbers $t_1$ and $t_2$, 
define $g(t_1,t_2)$ as $g(t_1,t_2)=D_F(\rho+t_1X_1\|\rho+t_2X_2)$.
Since $g(0,0)=F(1)=0$, Taylor's theorem and Proposition~\ref{propTaylorFdiv} (which is proved below) imply 
\begin{align*}
	g(t_1,t_2)
	&= \frac{1}{2}\Bigl( t_1^2J_\rho^f[X_1,X_1] + t_2^2J_\rho^f[X_2,X_2] - t_1t_2J_\rho^f[X_1,X_2] - t_2t_1J_\rho^f[X_2,X_1] \Bigr)\\
	&\quad+ o(\max\{t_1^2,t_2^2\})\\
	&= \frac{1}{2}J_\rho^f[t_1X_1-t_2X_2,t_1X_1-t_2X_2] + o(\max\{t_1^2,t_2^2\}).
\end{align*}
Substituting $\rho=\rho_0$, $X_k=\|\delta\rho_k\|^{-1}\delta\rho_k$, and $t_k=\|\delta\rho_k\|$, 
we obtain \eqref{eqTaylorFdiv}

\begin{proposition}\label{propTaylorFdiv}
	The above function $g$ satisfies 
	\begin{gather}
		\partial_1g(0,0) = \partial_2g(0,0) = 0,\label{eq47}\\
		\partial_1^2g(0,0) = J_\rho^f[X_1,X_1],\quad
		\partial_2^2g(0,0) = J_\rho^f[X_2,X_2],\label{eq48}\\
		\partial_1\partial_2g(0,0) = -J_\rho^f[X_1,X_2],\quad\text{and}\quad
		\partial_2\partial_1g(0,0) = -J_\rho^f[X_2,X_1],\nonumber
	\end{gather}
	where $f$ is defined in \eqref{eq05}.
\end{proposition}
\begin{proof}
	For small real numbers $t_1,t_2$ and $k\in[2]$, 
	define $g_k(t_1,t_2)$ as $g_k(t_1,t_2)=D_F(\rho+t_1X_k\|\rho+t_2X_k)$.
	By \cite[Theorems~2.8--2.9]{Lesniewski}, the equalities 
	\begin{equation}
		\partial_1\partial_2g(0,0) = -J_\rho^f[X_1,X_2]
		\quad\text{and}\quad
		\partial_1\partial_2g_k(0,0) = -J_\rho^f[X_k,X_k]
		\label{eq50}
	\end{equation}
	hold for every $k\in[2]$.
	Since $g$ is of class $C^2$ \cite[Theorem~2.4]{Lesniewski} and $J_\rho^f$ is symmetric \cite[Theorem~2.10]{Lesniewski}, 
	the equalities 
	\[
	\partial_1\partial_2g(0,0) = \partial_2\partial_1g(0,0)
	\quad\text{and}\quad
	J_\rho^f[X_1,X_2]=J_\rho^f[X_2,X_1]
	\]
	also hold.
	All we need is to show \eqref{eq47} and \eqref{eq48}.
	\par
	From the definitions of $D_F$, $g$, and $g_1$, it follows that 
	\begin{equation}
		g_1(t,t) = 0,\quad
		\partial_1g(t,0) = \partial_1g_1(t,0),\quad\text{and}\quad
		\partial_1^2g(t,0) = \partial_1^2g_1(t,0)
		\label{eq51}
	\end{equation}
	for every small $t$.
	Let $\Lambda$ be the quantum channel defined as $\Lambda(X)=(\Tr X)d^{-1}I_d$.
	The data processing inequality implies 
	\[
	g(t_1,t_2) \ge D_F(\Lambda(\rho+t_1X_1)\|\Lambda(\rho+t_2X_2))
	= D_F(d^{-1}I_d\|d^{-1}I_d)
	= F(1) = 0.
	\]
	Thus, the smooth function $t\mapsto g_1(t,t_2)$ achieves the minimum at $t=t_2$, 
	which implies $\partial_1g_1(t_2,t_2)=0$.
	This and \eqref{eq51} yield $\partial_1g(0,0)=\partial_1g_1(0,0)=0$.
	Moreover, by differentiating both sides of $\partial_1g_1(t,t)=0$, 
	we obtain 
	\[
	\partial_1^2g_1(0,0) + \partial_2\partial_1g_1(0,0) = 0.
	\]
	This, \eqref{eq51}, and \eqref{eq50} yield 
	\[
	\partial_1^2g(0,0)
	= \partial_1^2g_1(0,0)
	= -\partial_2\partial_1g_1(0,0)
	= J_\rho^f[X_1,X_1].
	\]
	We can also prove $\partial_2g(0,0)=0$ and $\partial_2^2g(0,0)=J_\rho^f[X_2,X_2]$ in the same way.
\end{proof}

\end{document}